\documentclass[a4paper,UKenglish,cleveref, autoref,thm-restate]{lipics-v2019}
\usepackage{color}
\usepackage{xcolor}
\definecolor{dkgreen}{rgb}{0,0.6,0}
\definecolor{gray}{rgb}{0.5,0.5,0.5}
\definecolor{lgray}{rgb}{0.85,0.85,0.85}
\definecolor{mauve}{rgb}{0.58,0,0.82}
\definecolor{red}{rgb}{1.0,0,0}

\usepackage{amsmath,amsfonts,amssymb}

\usepackage{pgfplots}
\usepackage{scalerel}
\usepackage{subcaption}
\usepackage{lcg}
\usepackage{xfp}
\usepackage{hyperref}
\usepackage{cleveref}
\usepackage[ruled,vlined,linesnumbered]{algorithm2e}
\usepackage{thmtools,thm-restate}
\usepackage{multirow}

\usetikzlibrary{fit}

\title{The Complexity of Graph Exploration Games}

\titlerunning{The Complexity of Graph Exploration Games} %TODO optional, please use if title is longer than one line

\author{Janosch Fuchs}{Department of Computer Science, RWTH Aachen University, Germany}{fuchs@algo.rwth-aachen.de}{https://orcid.org/0000-0003-3993-222X}{}

\author{Christoph Grüne}{Department of Computer Science, RWTH Aachen University, Germany}{gruene@algo.rwth-aachen.de}{https://orcid.org/0000-0002-7789-8870}{}

\author{Tom Janßen}{Department of Computer Science, RWTH Aachen University, Germany}{janssen@algo.rwth-aachen.de}{https://orcid.org/0000-0003-4617-3540}{}

\authorrunning{J. Fuchs, C. Grüne, T. Janßen}

\Copyright{Janosch Fuchs, Christoph Grüne and Tom Janßen}


 %TODO mandatory; please add comma-separated list of keywords

\category{} %optional, e.g. invited paper

\relatedversion{} %optional, e.g. full version hosted on arXiv, HAL, or other respository/website
\supplement{}%optional, e.g. related research data, source code, ... hosted on a repository like zenodo, figshare, GitHub, ...

\funding{This work is funded by the Deutsche Forschungsgemeinschaft (DFG, German Research Foundation) –- GRK 2236/1, WO 1451/2-1.}

\acknowledgements{}%optional

\nolinenumbers %uncomment to disable line numbering

\hideLIPIcs  %uncomment to remove references to LIPIcs series (logo, DOI, ...), e.g. when preparing a pre-final version to be uploaded to arXiv or another public repository

\EventEditors{John Q. Open and Joan R. Access}
\EventNoEds{2}
\EventLongTitle{42nd Conference on Very Important Topics (CVIT 2016)}
\EventShortTitle{CVIT 2016}
\EventAcronym{CVIT}
\EventYear{2016}
\EventDate{December 24--27, 2016}
\EventLocation{Little Whinging, United Kingdom}
\EventLogo{}
\SeriesVolume{42}
\ArticleNo{23}
\begin{document}

% !TEX root=lipics-paper.tex
\newcommand{\openNeighborhood}[2]{\ensuremath{N_{#1}(#2)}}
\newcommand{\closedNeighborhood}[2]{\ensuremath{N_{#1}[#2]}}
\newcommand{\NeighborhoodSubgraph}[2]{\ensuremath{G^N_{#1}[#2]}}

% general
\newcommand{\LSPACE}{\textsl{LOGSPACE}}
\newcommand{\NL}{\textsl{NL}}
\newcommand{\PTIME}{\textsl{PTIME}}
\newcommand{\NP}{\textsl{NP}}
\newcommand{\PSPACE}{\textsl{PSPACE}}

\newcommand{\indexSet}[2]{\ensuremath{~#1 \in \{1, \dots, #2\}}}

% generic problems and their online variant
\newcommand{\CombinatorialProblem}{\ensuremath{P^{C}}}
\newcommand{\CombinatorialGraphProblem}{\ensuremath{P^{CG}}}
\newcommand{\OnlineCombinatorialGraphProblem}{\ensuremath{P^{CG}_o}}
\newcommand{\VertexOrderingGraphProblem}{\ensuremath{P^{VOG}}}
\newcommand{\OnlineVertexOrderingGraphProblem}{\ensuremath{P^{VOG}_o}}

% Gadgets
\newcommand{\LiteralGadget}[1]{\ensuremath{G^U_{\ell_{#1}}}}
\newcommand{\NegLiteralGadget}[1]{\ensuremath{G^U_{\overline{\ell}_{#1}}}}
\newcommand{\VariableGadget}{\ensuremath{G_{var}}}
\newcommand{\ClauseGadget}{\ensuremath{G_{c}}}
\newcommand{\ExtensionGadget}{\ensuremath{G_{ext}}}
\newcommand{\IdGadget}{\ensuremath{G_{id}}}

% TO-DOS
\newcommand{\additions}[1]{{\leavevmode\color{blue}\textbf{ADD?: }#1}}
\newcommand{\todo}[1]{\textcolor{red}{\textbf{TODO: #1}}\\}
\newcommand{\todonn}[1]{\textcolor{red}{\textbf{TODO: #1}}}
\newcommand{\todolater}[1]{\textcolor{orange}{\textbf{TODO (later): #1}}\\}
\newcommand{\todolayout}[1]{\textcolor{violet}{\textbf{TODO (layout): #1}}\\}
\newcommand{\cit}{\textsuperscript{[citation needed]}}

% 3-sat variants
\newcommand{\sat}{\textsc{3-Satisfiability}}
\newcommand{\qsat}{\textsc{TQBF}}
\newcommand{\qsatgame}{\textsc{TQBF Game}}

% hamiltonian path
\newcommand{\hp}{\textsc{Hamiltonian Path}}
\newcommand{\dhp}{\textsc{Directed Hamiltonian Path}}
\newcommand{\uhp}{\textsc{Undirected Hamiltonian Path}}
\newcommand{\ohpg}{\textsc{Online Hamiltonian Path Game}}
\newcommand{\odhpg}{\textsc{Online Directed Hamiltonian Path Game}}
\newcommand{\ouhpg}{\textsc{Online Undirected Hamiltonian Path Game}}

\newcommand{\hc}{\textsc{Hamiltonian Cycle}}
\newcommand{\dhc}{\textsc{Directed Hamiltonian Cycle}}
\newcommand{\uhc}{\textsc{Undirected Hamiltonian Cycle}}
\newcommand{\ohcg}{\textsc{Online Hamiltonian Cycle Game}}
\newcommand{\odhcg}{\textsc{Online Directed Hamiltonian Cycle Game}}
\newcommand{\ouhcg}{\textsc{Online Undirected Hamiltonian Cycle Game}}

\newcommand{\tsp}{\textsc{Travelling Salesman Problem}}
\newcommand{\dtsp}{\textsc{Directed Travelling Salesman Problem}}
\newcommand{\utsp}{\textsc{Undirected Travelling Salesman Problem}}
\newcommand{\mtsp}{\textsc{Metric Travelling Salesman Problem}}
\newcommand{\btsp}{\textsc{Bottleneck Travelling Salesman Problem}}
\newcommand{\otsg}{\textsc{Online Travelling Salesman Game}}
\newcommand{\odtsg}{\textsc{Online Directed Travelling Salesman Game}}
\newcommand{\outsg}{\textsc{Online Undirected Travelling Salesman Game}}
\newcommand{\omtsg}{\textsc{Online Metric Travelling Salesman Game}}
\newcommand{\obtsg}{\textsc{Online Bottleneck Travelling Salesman Game}}

\newcommand{\stackercrane}{\textsc{Stacker-Crane}}
\newcommand{\oscg}{\textsc{Online Stacker-Crane Game}}
\newcommand{\rpp}{\textsc{Rural Postman Problem}}
\newcommand{\orpg}{\textsc{Online Rural Postman Game}}
\newcommand{\lp}{\textsc{Longest Path}}
\newcommand{\olpg}{\textsc{Online Longest Path Game}}
\newcommand{\lc}{\textsc{Longest Cycle}}
\newcommand{\olcg}{\textsc{Online Longest Cycle Game}}

% s-t path
\newcommand{\stp}{\textsc{s-t Path}}
\newcommand{\ustp}{\textsc{Undirected s-t Path}}
\newcommand{\dstp}{\textsc{Directed s-t Path}}
\newcommand{\ostpg}{\textsc{Online s-t Path Game}}
\newcommand{\ostupg}{\textsc{Online Undirected s-t Path Game}}
\newcommand{\ostdpg}{\textsc{Online Directed s-t Path Game}}

\newcommand{\stt}{\textsc{s-t Trail}}
\newcommand{\ustt}{\textsc{Undirected s-t Trail}}
\newcommand{\dstt}{\textsc{Directed s-t Trail}}
\newcommand{\osttg}{\textsc{Online s-t Trail Game}}
\newcommand{\ostutg}{\textsc{Online Undirected s-t Trail Game}}
\newcommand{\ostdtg}{\textsc{Online Directed s-t Trail Game}}

\newcommand{\stw}{\textsc{s-t Walk}}
\newcommand{\ustw}{\textsc{Undirected s-t Walk}}
\newcommand{\dstw}{\textsc{Directed s-t Walk}}
\newcommand{\ostwg}{\textsc{Online s-t Walk Game}}
\newcommand{\ostuwg}{\textsc{Online Undirected s-t Walk Game}}
\newcommand{\ostdwg}{\textsc{Online Directed s-t Walk Game}}
\newcommand{\ostuwwg}{\textsc{Online Undirected Shortest s-t Walk Game}}
\newcommand{\ostdwwg}{\textsc{Online Directed Shortest s-t Walk Game}}

\newcommand{\opwfpg}{\textsc{Online Path With Forbidden Pairs Game}}
\newcommand{\pwfp}{\textsc{Online Path With Forbidden Pairs Game}}
\newcommand{\csp}{\textsc{Constrained Shortest Path}}
\newcommand{\ocspg}{\textsc{Online Constrained Shortest Path Game}}
\newcommand{\tdp}{\textsc{Two Disjoint Path}}
\newcommand{\otdpg}{\textsc{Online Two Disjoint Path Game}}
\newcommand{\vdp}{\textsc{Vertex Disjoint Path}}
\newcommand{\ovdpg}{\textsc{Online Vertex Disjoint Path Game}}

\maketitle

\begin{abstract}
	% !TEX root=lipics-paper.tex
Graph Exploration problems ask a searcher to explore an unknown environment.
The environment is modeled as a graph, where the searcher needs to visit each vertex beginning at some vertex.
Treasure Hunt problems are a variation of Graph Exploration, in which the searcher needs to find a hidden treasure, which is located at a designated vertex.

Usually these problems are modeled as online problems, and any online algorithm performs poorly because it has too little knowledge about the instance to react adequately to the requests of the adversary.
Thus, the impact of a priori knowledge is of interest.
One form of a priori knowledge is an unlabeled map, which is an isomorphic copy of the graph.
We analyze Graph Exploration and Treasure Hunt problems with an unlabeled map that is provided to the searcher.
For this, we formulate decision variants of both problems by interpreting the online problems as a game between the online algorithm (the searcher) and the adversary.
The map, however, is not controllable by the adversary.
The question is whether the searcher is able to explore the graph completely or find the treasure for all possible decisions of the adversary.

We analyze these games in multiple settings, with and without costs on the edges, on directed and undirected graphs and with different constraints (allowing multiple visits to vertices or edges) on the solution.
We prove \PSPACE-completeness for most of these games.
Additionally, we analyze the complexity of related problems that have additional constraints on the solution.

\end{abstract}

\newpage

% !TEX root=lipics-paper.tex

% !TEX root=lipics-paper.tex
\section{Introduction}

Graph Exploration problems model situations in which a searcher, like an autonomous robot, has to explore an environment and solve a task.
Among those tasks are finding the shortest path to a designated point, which is referred to as the Treasure Hunt problem, or exploring the whole environment with minimal resource consumption, which is referred to as the Online Traveling Salesman problem.
Thereby, the searcher does not know the environment at the beginning and only obtains local information during exploration. 

Typically, Graph Exploration is modeled as an online problem on a graph, that is fixed by the adversary before the online computation starts.
The searcher is positioned at a vertex and then the labels of its neighborhood are revealed together with the corresponding incident edges.
Kalyanasundaram and Pruhs defined this model as \emph{fixed graph scenario} \cite{DBLP:journals/tcs/KalyanasundaramP94}.
Based on the overall obtained knowledge, the online algorithm has to irrevocably decide along which edge the searcher moves.
For a worst-case analysis a malicious adversary is presupposed which controls the revelation process and creates the input.
The goal of the adversary is to minimize the performance of the online algorithm.

While the online algorithm moves the searcher, the adversary creates the graph and chooses the vertices that are revealed.
Therefore, the adversary is able to tailor the instance in his favor to the decisions of any online algorithm.
To overcome this asymmetry, different extensions of the online setting exist, in which the online algorithm is equipped with a priori knowledge.
Throughout this paper, we introduce an unlabeled map, which is an isomorphic copy of the input graph.
Thus, the input graph is not constructed by the adversary and only the revelation order of the vertices is determined by the adversary.

The connection between the online algorithm and the adversary is analogous to two players in an asymmetric two-player game \cite{DBLP:conf/iwoca/BohmV16,DBLP:conf/sofsem/FuchsGJ24,DBLP:conf/icalp/PapadimitriouY89}.
The input graph can be considered as the game board.
A \emph{turn} of the game consist of a \emph{move of the adversary} followed by \emph{move of the online algorithm}.
Specifically, the adversary reveals the neighborhood of the vertex $v$, on which the searcher is positioned, by revealing the labels of the neighbors and the edges connecting them to $v$.
The labels are recognizable by the online algorithm later in the game.
Thereafter, the online algorithm makes a move by choosing an incident edge of $v$ to move the searcher along.
The problem is to decide whether the online algorithm has a \emph{winning strategy}, that is, it can compute a feasible sequence of vertices, for all possible moves of the adversary.

\paragraph*{Related Work}
Among the preliminary work on online path problems are the Online Traveling Salesman Problem and the Canadian Traveler Problem.
Papadimitriou and Yannakakis~\cite{DBLP:conf/icalp/PapadimitriouY89} introduced the problem of finding a shortest path in a graph with edge cost uncertainties, which are revealed when the searcher is positioned at an incident edge, as the Canadian Traveler Problem.
Further results on this problem are discussed by Bar-Noy and Schieber~\cite{DBLP:conf/soda/Bar-NoyS91}.
Kalyanasundaram and Pruhs~\cite{DBLP:journals/tcs/KalyanasundaramP94} introduce the online version of the Traveling Salesman Problem under the fixed graph scenario, which is later referred as Graph Exploration.
They also present an algorithm that yields a $16$-approximation for undirected planar graphs.
Foerster and Wattenhofer~\cite{DBLP:conf/opodis/ForsterW12} used the same model as in Graph Exploration to analyze the Treasure Hunt Problem.
Additionally they provide lower and upper bounds on the competitivity of Graph Exploration on directed graphs.
Bounds in the undirected case for Graph Exploration are provided by Megow et al.~\cite{DBLP:journals/tcs/MegowMS12}.
They also show that the $16$-approximation of Kalyanasundaram and Pruhs~\cite{DBLP:journals/tcs/KalyanasundaramP94} extends to graphs of bounded genus, where the competitive ratio increases linearly with the genus.
This result was further improved and extended to graphs with excluded minors by Baligács et al. \cite{DBLP:conf/esa/BaligacsDHS23}.
Furthermore, there are constrained variations of the Graph Exploration problem that limit the ability of the searcher.
Duncan et. al.~\cite{DBLP:journals/talg/DuncanKK06} analyze Graph Exploration where the searcher is tied to the starting point with a tether of fixed length or has a limited fuel tank.
They give upper and lower bounds for these settings.

Another branch of Graph Exploration surveys the influence of additional information on the performance of the searcher.
A subset of those variations includes some form of a map.
Panaite and Pelc~\cite{DBLP:journals/networks/PanaiteP00} focus on a setting, where the searcher has either a labeled map, a labeled copy of the graph with an additional sense of direction, or an unlabeled map, an isomorphic copy of the graph, and compare these models. 
Furthermore, Dessmark and Pelc~\cite{DBLP:journals/tcs/DessmarkP04} use a similar model where the searcher has an unlabeled map and either knows where it starts on this map (anchored map) or not (unanchored map).
Additionally, maps are also of interest for the Treasure Hunting Problem. Bouchard et al.~\cite{DBLP:journals/networks/BouchardLP22} analyze the performance gain of using different forms of maps.
Instead of using a model of a map, an abstract and general form of information may be used as well, the so-called advice model.
The advice is provided as a binary string, whereby the advice complexity is the number of used bits.
Dobrev et al.~\cite{DBLP:conf/sirocco/DobrevKM12} give a lower bound of $\Omega(|V|\log(|V|))$ on the advice complexity when the algorithm has to compute an optimal solution and present an algorithm using linear advice and achieving a constant competitive ratio of $6$.
Böckenhauer et al.~\cite{DBLP:journals/iandc/BockenhauerFU22} show that $\mathcal{O}(|E|)$ advice bits are sufficient to optimally explore any graph.
Besides, Komm~\cite{DBLP:conf/sirocco/KommKKS15} et al. analyze the Treasure Hunt Problem with advice.
At last, a new branch uses prediction models as source of information.
Eberle et al.~\cite{DBLP:conf/aaai/EberleLMNS22} consider a learning prediction framework with a bounded error to potentially robustify existing algorithms.

With the work on the Canadian Traveler Problem, Papadimitriou and Yannakakis~\cite{DBLP:conf/icalp/PapadimitriouY89} also introduced online graph games with a map.
The task is to find a shortest $s$-$t$-path where the edge costs are chosen by the adversary.
They show the \PSPACE-completeness to decide whether there is an $r$-competitive strategy for traversing the graph, where $r$ is a given ratio.
This work is also continued by Bar-Noy and Schieber~\cite{DBLP:conf/soda/Bar-NoyS91} on different variations, where the $k$-Canadian Traveler Problem remains \PSPACE-complete.
Additionally, Böhm and Vesel\'{y}~\cite{DBLP:conf/iwoca/BohmV16} show the Online Chromatic Number problem to be \PSPACE-complete.
In there, an unlabeled map is provided to the online algorithm.
In a similar setting, Fuchs et al. \cite{DBLP:conf/sofsem/FuchsGJ24} build a reduction framework which can be applied to graph problems that search for a subset of vertices such as vertex cover, independent set or dominating set.
With that they show that online games based on these problems are \PSPACE-complete.
A complexity analysis on a broader set of \PSPACE-hard combinatorial games can be found in Fraenkel and Goldschmidt's survey~\cite{DBLP:journals/jct/FraenkelG87}.

\paragraph*{Contribution}
We analyze the complexity properties of Graph Exploration problems by taking up the ideas by Papadimitriou and Yannakakis~\cite{DBLP:conf/icalp/PapadimitriouY89}, Böhm and Vesel\'{y}~\cite{DBLP:conf/iwoca/BohmV16} as well as Fuchs et al. \cite{DBLP:conf/sofsem/FuchsGJ24}.
That is, we introduce online games variants of Graph Exploration problems that include an unlabeled map of the graph that the online algorithm can use.
On the one hand, we define and analyze the Online Traveling Salesman Game, which is the online game version of the original Graph Exploration problem defined by Kalyanasundaram and Pruhs~\cite{DBLP:journals/tcs/KalyanasundaramP94}.
It asks whether an online algorithm is able to find a Hamiltonian cycle of small weight in a given graph for all possible reveal decisions of the adversary while having an unlabeled map.
On the other hand, we define the online game version of the Treasure Hunt problem.
It asks whether an online algorithm is able to find an s-t-path in a given graph for all possible reveal decisions of the adversary while having an unlabeled map.

Furthermore, we analyze variants of both problems:
Besides merely asking for the existence of a path or cycle, an additional number $k \in \mathbb{N}$ is introduced, limiting the length of the solution.
Additionally, we consider versions of both problems, in which we relax the path constraint to be a trail or a walk as well as constrained versions of these problems such as the metric version of the online traveling salesman game. 
We show that nearly all of the above mentioned problems are \PSPACE-complete.
The other problems degenerate to simple offline problems such as \ostuwg, which is the \LSPACE-complete problem \textsc{UstCon} \cite{DBLP:journals/jacm/Reingold08}.

\paragraph*{Paper Summary}
In \Cref{sec:preliminaries}, we define preliminary terms including complexity theoretic concepts and the online game setting.
In \Cref{sec:path}, we analyze the \ostpg{} as well as variants in directed and undirected graphs with and without edge costs.
In \Cref{sec:hampath}, we examine the results on \ohpg{} and the related variants in directed and undirected graphs with and without costs.
Then, variations of the classical \stp{} and \hp{} are studied in the online game context in \Cref{sec:additional}.
At last in \Cref{sec:conclusion}, we conclude the paper and present remaining open problems.

% !TEX root=lncs-paper.tex
\section{Preliminaries}
\label{sec:preliminaries}
As usual, we define a walk as a sequence of connected edges.
A trail is a walk where all edges are distinct, and a path is a trail such that no vertex occurs more than once.
We also refer to $s$-$t$ walks (resp. trails, paths) to indicate the two endpoints of the walks (resp. trails, paths).
With $N(v)$ we refer to the open neighborhood of vertex $v$.

\paragraph*{Search Sequences}
A search sequence is a (valid) solution to an instance of a Graph Exploration problem.
Intuitively, a search sequence is a walk, which does not contain cycles consisting only of vertices that have occured in the same walk before.  
That is in every cyclic subwalk, a previously non-visited vertex has to be included in the walk.

\begin{definition}[Search Sequence]\label{def:searchSequence}
	For a graph $G$, a \emph{search sequence} is a sequence of arcs or edges $e_1, \dots, e_{n-1}$ in $G$ for which there is a sequence $v_1, \dots, v_n$ of vertices in $G$, such that $e_i = \{v_i, v_{i+1}\}$.
	Furthermore, for all subsequences $v_j, \dots, v_k$ of $v_1, \dots, v_n$ with $v_j = v_k$ it holds that $\{v_1, \dots, v_{j-1}\} \subsetneq \{v_1, \dots, v_k\}$.
	If $v_1 = v_{n}$, we call $S$ a \emph{cyclic search sequence}.
	The cost of a search sequence $S = e_1, \dots, e_{n-1}$ is defined by $cost(S)=\sum_{e \in S} cost(e)$, if the edges have costs assigned, and $cost(S)= |S|$ otherwise.
\end{definition}

In the online setting, the search sequence is determined by the moves of the adversary as well as the moves of the online algorithm.
In each step, the online algorithm is located at some vertex $v$ and chooses one of the vertices from $N(v)$ as the target and moves itself towards it.
Then the adversary reveals the neighborhood of the target.
That is, the adversary reveals the labels of all neighboring vertices $n \in N(v)$ as well as all edge/arc weights.
We refer to this model as the \emph{neighborhood reveal model}.
Throughout the paper, we call  vertices to which the online algorithm moved before \emph{visited} and vertices which are revealed for the first time \emph{new}.
Furthermore, we call non-visited vertices which are revealed a subsequent time \emph{known}.

The online algorithm may be restricted to different variations of search sequences.
By definition, a search sequence has to be a walk in the graph.
We also consider problems that restrict the search sequences to trails or paths.
While trails and paths are polynomially bounded in their length by the size of the graph, this is generally not the case for walks.
However, a search sequence may not contain any cycle that does not visit previously unvisited vertices.
This does not restrict the online algorithm, since traversing a cycle of only visited vertices does not reveal any new vertices and puts the online algorithm back in the position it was before.
Thus the length of a search sequence is always polynomially bounded in the size of the input graph.

\paragraph*{Complexity Theory}
We define a decision problem to be a subset of $\{0,1\}^*$.
For two decision problems $A$ and $B$, we say that $A$ is polynomially reducible to $B$, if there is a function $f: \{0,1\}^* \rightarrow \{0,1\}^*$ computable in polynomial time such that $x \in A$ if and only if $f(x) \in B$.
The class \PSPACE{} is given by all decision problems that can be decided by a deterministic Turing machine with polynomial space.
As for \NP, a problem is called \PSPACE-hard, if any other problem in \PSPACE{} can be reduced to it by a polynomial reduction.
A problem that is both contained in \PSPACE{} and is \PSPACE-hard, is also called \PSPACE-complete.
The canonical \PSPACE-complete problem is \textsc{True Quantified Boolean Formula} \cite{DBLP:conf/stoc/StockmeyerM73} or \qsat{} for short.
For this paper, the game version of \qsat{} -- \qsatgame{} -- is of most interest.

This game is played by two players: the $\exists$-player and the $\forall$-player.
The $\exists$-player controls all $\exists$-quantified variables and the $\forall$-player controls all $\forall$-quantified variables in the order of quantification.
That is, a turn consists of a move of the $\exists$-player followed by a move of the $\forall$-player, in which they decide the assignment of their variable(s).
The $\exists$-player wins if and only if $\varphi(X_1, \ldots, X_n)$ is satisfied with the assignments of both players.

\begin{definition}[\qsatgame]\label{prelim:def:tqbfgame}\hfill\\
	\textbf{Given:} A fully quantified Boolean formula $Q_1 X_1 \ldots Q_n X_n \varphi(X_1, \ldots, X_n)$ with $Q_i \in \{\exists, \forall\}$ for $i \in \{1, \ldots, n\}$. \\
	\textbf{Question:} Does the $\exists$-player have a winning strategy?
\end{definition}

Deciding whether the $\exists$-player has a winning strategy is \PSPACE-complete by a simple reduction from \qsat.
W.l.o.g. we assume $\varphi$ to be in CNF.
Furthermore, we assume clauses to only contain three literals for simplicity, but our constructions also extend to any number of literals per clause.

\paragraph*{Online Search Sequence Games} 
A search sequence problem $P^{SSP}$ has a graph $G$ as input and the feasible solutions are a subset of all search sequences in $G$.
Examples for such problems are \textsc{s-t Path} and \textsc{Hamiltonian Path}.
For any problem $P^{SSP}$, we define an online game version.
\begin{definition}[Online Search Sequence Game]\label{prelim:def:onlinegame}\hfill\\
	\textbf{Given:} A graph $G$ and possibly start and/or end vertices. \\
	\textbf{Question:} Does an online algorithm exist, that finds a valid search sequence in $G$ (as defined by $P^{SSP}$) for all strategies of the adversary in the neighborhood reveal model, while the online algorithm knows an unlabeled map of $G$?
\end{definition}
We also refer to this problem as $P^{SSP}_O$.
Since our definition of search sequences implies them having a length polynomial in the size of the input graph (as argued above), we obtain the following theorem.

\begin{restatable}{theorem}{prelimthmpspaceContainment}\label{prelim:thm:pspaceContainment}
	If $P^{SSP} \in NP$, then $P_O^{SSP} \in \PSPACE$.
\end{restatable}
% !TEX root=lipics-paper.tex
\begin{proof}
	The instance graph is encoded in linear space.
	The solution (sequence of edges) is encoded in at most polynomial space because the base problem $P^{SSP}$ is in \NP.
	The number of turns is different for the type of the problem.
	\begin{itemize}
		\item For path problems, the number of turns is bounded by the number of vertices ($|V|$).
		\item For trail problems, the number of turns is bounded by the number edges and arcs ($|E| + |A|$).
		\item For walk problems, the definition of a search sequence bounds the number of turns by $|V|^2$.
	\end{itemize}
	Thus, the number of turns is polynomial in the input.
	For each turn, the currently revealed graph is stored as well as the current solution.
	This is polynomial in the input.
	Thus, the used space is overall polynomial for each turn.
	Consequently, the problem $P_O^{SSP}$ is in \PSPACE.
\end{proof}

% !TEX root=lipics-paper.tex
\section{Path Problems}
\label{sec:path}
The first class of problems that we analyze are $s$-$t$ path problems.
The online versions of these problems can be interpreted as a Treasure Hunt problem.
We start our complexity analysis with the \ostupg{}.

\begin{definition}[\ostupg]\label{path:def:ostupg}\hfill\\
	\textbf{Given:} An undirected graph $G = (V, E)$, and two vertices $s, t \in V$. \\
	\textbf{Question:} Does an online algorithm exist, that finds a $s$-$t$ path in $G$ for all strategies of the adversary in the neighborhood reveal model, while the online algorithm knows an unlabeled map of $G$?
\end{definition}

We show that \ostupg{} is \PSPACE-complete and derive further results on variations which include $s$-$t$ path, $s$-$t$ trail and $s$-$t$ walk on directed and undirected graphs.
Additionally, we survey the online versions of constrained path problems.

\begin{restatable}{theorem}{paththmuPathPspaceHard}\label{path:thm:uPathPspaceHard}
	\ostupg{} is \PSPACE-complete.
\end{restatable}

\paragraph*{Reduction Overview}
We show the \PSPACE-hardness for \ostupg{} by a reduction from \qsatgame{}.
However, when considering games based on online problems, the online algorithm always chooses the next vertex.
We can still model choices of the adversary though, by letting the online algorithm choose between two vertices it cannot distinguish.
This way, the online algorithm is able to decide the truth assignment of $\exists$-variables, and the adversary is able to decide the truth assignment of $\forall$-variables.

The following reduction is loosely adapted from the reduction of Li et. al. \cite{DBLP:journals/dam/LiMS90}.
The variable gadget of our reduction essentially consists of two paths, which end in the same vertex.
One path corresponds to assigning the variable the value true, the other the value false.
The clause gadget consists of two disconnected vertices, connected by one path for each literal they contain, where the first vertex of each path has an edge to the respective variable gadget.
If for a clause at least one variable gadget is set to a value that satisfies the clause, then the online algorithm can identify one of the paths to traverse the clause gadget.
The $s$-$t$ path, the online algorithm needs to find, starts in $s$, then traverses all variable gadgets (using only one of the two paths), then traverses all clause gadgets, and finally reaches $t$.

Two important parts of our reduction that use the online nature of the game are additional edges to reveal vertices, making them recognizable for later decisions, and \emph{sinks}.
A sink being attached to a vertex $v$ means that there is a vertex connected to $v$ by an edge but to no other vertices.
The purpose of a sink is to prevent the online algorithm from choosing a new neighbor of $v$, as that allows the adversary to trap it in the sink.
Edges that reveal vertices are usually used together with sinks, to prevent the online algorithm from traversing them, but still allow it to recognize a vertex later.

\paragraph*{Variable Gadget Overview}

The variable gadget for variable $x_i$ roughly consists of four parts, using the same numbering as in \Cref{st:fig:variable}:
\begin{enumerate}
	\item There are two paths, one corresponding to assigning the variable the value true, and the other corresponding to the value false.
	In \Cref{st:fig:variable}, these paths start at the vertices $v^i_1$ and $w^i_1$, and meet at $u^i_6$.
	The variable is set to true by choosing the path corresponding to true. 
	This reveals vertices in the clauses that are satisfied, helping the algorithm traversing those clause gadgets later.
	\item There are vertices simulating the decision for the variable assignment of the online algorithm or adversary, depending on whether $x_i$ is $\exists$-quantified or $\forall$-quantified.
	In \Cref{st:fig:variable}, this is done by the vertices $v^i$ and $w^i$.
	The green edge only exists if $x_i$ is $\exists$-quantified, and reveals which vertex corresponds to which path.
	On the other hand, if $x_i$ is $\forall$-quantified, the online algorithm cannot distinguish $v^i$ and $w^i$.
	Only after choosing one of the two, the online algorithm learns the truth assignment from the map, as the map shows whether the path corresponding to true (resp. false) has clauses attached first.
	This is necessary, so the online algorithm can make its choices for $\exists$-variables dependent on the choices for $\forall$-variables of the adversary.
	\item There is a path of vertices whose purpose is to reveal other parts of the gadget, so the online algorithm can recognize them later.
	As described above, this includes revealing the truth assignments that correspond to the two paths, but also further vertices on the paths, indicated by the cyan and dot-dot-dashed edges in \Cref{st:fig:variable}.
	This is used together with sinks to prevent the online algorithm from unwanted behavior.
	\item There is a small gadget that ensures the online algorithm correctly finds its way through the path that reveals the later parts of the gadget.
	Furthermore this gadget enforces that the algorithm traverses the entire path of Part 3, to prevent it from producing a mixed variable assignment by using it to switch between the two paths of Part 1.
\end{enumerate}
% !TEX root=lipics-paper.tex
\begin{figure}[t]
	% colors
	\def\sinkcolor{red}
	\def\revealcolorid{blue}\def\revealcoloriddot{dashed}
	\def\wincolor{dkgreen}\def\wincolordot{densely dash dot}
	\def\revealcolortfe{dkgreen}\def\revealcolortfedot{densely dash dot}
	\def\revealcolorva{cyan}\def\revealcolorvadot{densely dash dot dot}
	\def\revealcolorclause{mauve}\def\revealcolorclausedot{dotted}
	\def\clausecolor{orange}\def\clausecolordot{densely dotted}

	\centering
	\resizebox{0.82\textwidth}{!}{
		\begin{tikzpicture}[scale=0.4,
			node/.style = {shape=circle, draw, inner sep=0pt, minimum size=0.25cm},
			textnode/.style = {shape=circle, draw, inner sep=0pt, minimum size=0.4cm},
			sink/.style = {shape=circle, draw=\sinkcolor, inner sep=0pt, minimum size=0.1cm},
			box/.style = {rectangle, fill=gray!20, rounded corners, fill opacity=1, inner sep=5pt}]
			%\clip (-2,6.2) rectangle (27,-6.2);
			% --------------------------------------------------------------------------------
			% boxes around parts
			% --------------------------------------------------------------------------------
			% part 1
			% adjusted y-coords to make box heights the same
			\node (x1t0copy) at (11, 3) {};
			\node (x1f0copy) at (11, -3) {};
			\node (x1e2copy) at (23, 0) {};
			\node (sx1t6copy) at (21, 3.8) {};
			\node (sx1f2copy) at (15, -3.8) {};
			\node[box, fit=(x1t0copy)(x1f0copy)(x1e2copy)(sx1t6copy)(sx1f2copy)] (box1) {};
			\node[below right, inner sep=3pt] at (box1.north west) {1.};
			% part 2
			% adjusted y-coords to make box heights the same
			\node (x1tcopy) at (9, 3.8) {};
			\node (x1fcopy) at (9, -3.8) {};
			\node (x1dcopy) at (8, 0) {};
			\node[box, fit=(x1tcopy)(x1fcopy)(x1dcopy)] (box2) {};
			\node[below right, inner sep=3pt] at (box2.north west) {2.};
			% part 3
			% adjusted y-coords to make box heights the same
			\node (id21copy) at (5.5, -3.8) {};
			\node (id29copy) at (5.5, 3.8) {};
			\node (sid21copy) at (6, 3.8) {};
			\node[box, fit=(id21copy)(id29copy)(sid21copy)] (box3) {};
			\node[below right, inner sep=3pt] at (box3.north west) {3.};
			% part 4
			% adjusted y-coords to make box heights the same
			\node (id11copy) at (1.5, 3.8) {};
			\node (id20copy) at (3.5, -3.8) {};
			\node[box, fit=(id11copy)(id20copy)] (box4) {};
			\node[below right, inner sep=3pt] at (box4.north west) {4.};
			
			% --------------------------------------------------------------------------------
			% t/f labels
			% --------------------------------------------------------------------------------
			\node at (21.5, 3) {true};
			\node at (21.5, -3) {false};
			
			% --------------------------------------------------------------------------------
			% ID gadget of ID gadget 
			% --------------------------------------------------------------------------------
			% triangle
			\node[textnode] (id11) at (1.5, 0) {\tiny$u^i_1$};
			\node[textnode] (id12) at (2.5, 1) {\tiny$u^i_3$};
			\node[textnode] (id13) at (2.5, -1) {\tiny$u^i_2$};
			\path[-] (id11) edge (id12);
			\path[-] (id11) edge (id13);
			\path[-, ultra thick] (id12) edge (id13);
			% sinks
			\node[sink] (sid12) at (2.5, 2) {};
			\node[sink] (sid13) at (2.5, -2) {};
			\path[-, draw=\sinkcolor] (id12) edge (sid12);
			\path[-, draw=\sinkcolor] (id13) edge (sid13);
			
			% --------------------------------------------------------------------------------
			% ID gadget
			% --------------------------------------------------------------------------------
			% revealing nodes
			\node[textnode] (id20) at (3.5, 0) {\tiny$u^i_4$};
			\node[textnode] (id21) at (6, 2.25) {\tiny$p^i_1$};
			\path[-, ultra thick] (id20) edge (id21);
			\foreach \x in {2,3,4} {
				\pgfmathtruncatemacro{\y}{\x-1}
				\node[textnode] (id2\x) at (6, -1.5*\x+3.75) {\tiny$p^i_\x$};
				\path[-, ultra thick] (id2\y) edge (id2\x);
			}
			% extra edges that reveal the path
			\foreach \x in {2,3,4}
				\path[-, draw=\revealcolorid, \revealcoloriddot] (id20) edge (id2\x);
			% sink at start node
			\node[sink] (sid20) at (3.5, -1) {};
			\path[-, draw=\sinkcolor] (id20) edge (sid20);
			% sinks along the path
			\foreach \x in {1,...,4} {
				\node[sink] (sid2\x) at (5.2, -1.5*\x+3.75) {};
				\path[-, draw=\sinkcolor] (id2\x) edge (sid2\x);
			}
			% connections from IDID gadget
			\path[-] (id12) edge (id20);
			\path[-] (id13) edge (id20);
			\path[-, draw=\revealcolorid, \revealcoloriddot] (id12) edge (id21);
			
			% --------------------------------------------------------------------------------
			% variable gadget
			% --------------------------------------------------------------------------------
			% decision and revealing part
			\node[textnode] (x1d) at (8, 0) {\tiny$u^i_5$};
			\node[textnode] (x1t) at (9, 2.25) {\tiny$v^i$};
			\node[textnode] (x1t0) at (11, 3) {\tiny$v^i_1$};
			\node[textnode] (x1f) at (9, -2.25) {\tiny$w^i$};
			\node[textnode] (x1f0) at (11, -3) {\tiny$w^i_1$};
			\path[-] (x1d) edge (x1t) (x1t) edge (x1t0);
			\path[-] (x1d) edge (x1f) (x1f) edge (x1f0);
			% paths corresponding to true and false
			\foreach \x in {1,...,4} {
				\pgfmathtruncatemacro{\y}{\x-1}
				\pgfmathtruncatemacro{\z}{\x+1}
				\node[textnode] (x1t\x) at (11+2*\x, 3) {\tiny$v^i_\z$};
				\path[-] (x1t\y) edge (x1t\x);
				
				\node[textnode] (x1f\x) at (11+2*\x, -3) {\tiny$w^i_\z$};
				\path[-] (x1f\y) edge (x1f\x);
			}
			% end of gadget
			\node[textnode] (x1e1) at (21, 0) {\tiny$u^i_6$};
			\node[textnode] (x1e2) at (23, 0) {\tiny$u^i_7$};
			\path[-] (x1e1) edge (x1t4) edge (x1f4);
			\path[-, ultra thick] (x1e1) edge (x1e2);
			% sinks 
			\node[sink] (sx1t1) at (13, 2) {};
			\path[-, draw=\sinkcolor] (x1t1) edge (sx1t1);
			\node[sink] (sx1f3) at (17, -2) {};
			\path[-, draw=\sinkcolor] (x1f3) edge (sx1f3);
			\node[sink] (sx1e1) at (19.8, 0) {};
			\path[-, draw=\sinkcolor] (x1e1) edge (sx1e1);
			% clause stubs
			\node (c1) at (13, 5.7) {clauses containing $x_i$};
			\node (c11) at (12, 5.5) {};
			\node (c12) at (13, 5.5) {};
			\node (c13) at (14, 5.5) {};
			\path[-, draw=\clausecolor, \clausecolordot] (x1t1) edge (c11) edge (c12) edge (c13);
			\node (c2) at (17, -5.7) {clauses containing $\overline{x}_i$};
			\node (c21) at (16, -5.5) {};
			\node (c22) at (17, -5.5) {};
			\node (c23) at (18, -5.5) {};
			\path[-, draw=\clausecolor, \clausecolordot] (x1f3) edge (c21) edge (c22) edge (c23);
			% edges connecting ID gadget to variable decision
			\path[-, draw=\revealcolorid, \revealcoloriddot] (id21) edge (x1d);
			\path[-, ultra thick] (id24) edge (x1d);
			
			% --------------------------------------------------------------------------------
			% reveal edges
			% --------------------------------------------------------------------------------
			% edges revealing t/f decision
			\path[-, draw=\revealcolortfe, \revealcolortfedot] (id22) edge (x1t);
			% edges revealing parts of the ID gadget (always)
			\path[-, draw=\revealcolorva, \revealcolorvadot, out=60, in=240, out distance=4.5cm, in distance=3cm] (id23) edge (x1t2);
			\path[-, draw=\revealcolorva, \revealcolorvadot, out=350, in=120, out distance=4cm, in distance=3cm] (id23) edge (x1f4);
			\path[-, draw=\revealcolorva, \revealcolorvadot, out=0, in=225, out distance=7cm, in distance=1cm] (id23) edge (x1e2);
			
			% --------------------------------------------------------------------------------
			% connection to remaining graph
			% --------------------------------------------------------------------------------
			\node[textnode] (s) at (-0.5, 0) {$s$};
			\node[textnode] (t) at (25, 0) {$t$};
			\path[-, dashed] (s) edge (id11) (x1e2) edge (t);
		\end{tikzpicture}
	}
	\caption{Variable gadget for the reduction from \qsatgame{} to \ostpg{} for a variable $x_i$. 
		The thick solid black edges have to be used by the algorithm, while the thin solid black edges are optional.
		Any non-solid colored edge cannot be used, and only exists for later recognition of the vertex it connects to. 
		These properties are ensured by the sinks (red) that allow the adversary to trap the algorithm whenever it uses a non-solid non-black edge.
	}
	\label{st:fig:variable}
\end{figure}
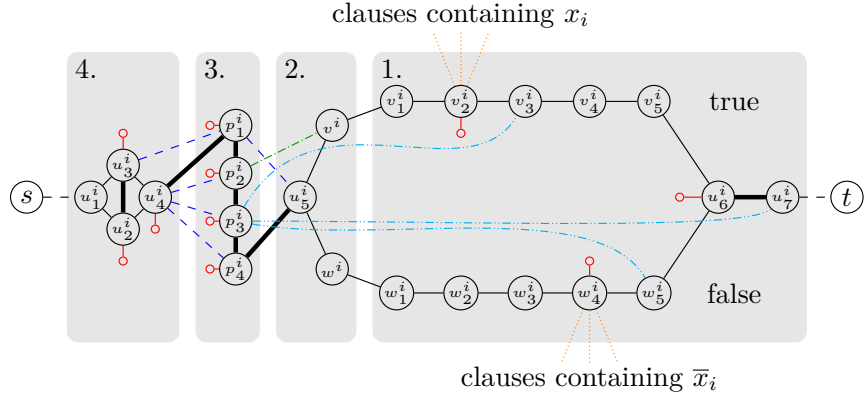 

\paragraph*{Formal Definition of the Variable Gadget}\label{path:def:varGadget}
	\begin{enumerate}
		\item In Part 1 there are vertices $v^i_j, w^i_j$ for $j \in \{1, \dots, 5\}$, where the vertices $v^i_j$ (resp. $w^i_j$) form a path.
			Further, there are vertices $u^i_6, u^i_7$ and edges $\{v^i_5, u^i_6\}, \{w^i_5, u^i_6\}, \{u^i_6, u^i_7\}$. 
			The vertices $u^i_6$, $v^i_2$ and $w^i_4$ have a sink attached.
			The vertex $v^i_2$ (resp. $w^i_4$) has an edge to the vertex $c^j_i$ (resp. $c^k_i$) of the $j$-th (resp. $k$-th) clause containing $x_i$ as a positive (resp. negative) literal.
		\item Part 2 consists of $u^i_5, v^i$ and $w^i$.
			The vertices $u^i_5, v^i$ and $v^i_1$ (resp. $u^i_5, w^i$ and $w^i_1$) form a path.
		\item In Part 3 there are vertices $p^i_j$, for $j \in \{1, \dots, 4\}$, forming a path, and each vertex has a sink attached. 
			Further, there are edges $\{p^i_3, v^i_3\}$, $\{p^i_3, w^i_5\}$, $\{p^i_3, u^i_7\}$, $\{p^i_1, u^i_5\}$ and $\{p^i_4, u^i_5\}$.
			If $x_i$ is $\exists$-quantified, there is the edge $\{p^i_2, v^i\}$. 
		\item Part 4 consists of four vertices, $u^i_1, u^i_2, u^i_3, u^i_4$, with $u^i_2, u^i_3$ and $u^i_4$ each having a sink attached.
			Further, there are edges $\{u^i_1, u^i_2\}$, $\{u^i_1, u^i_3\}$, $\{u^i_2, u^i_3\}$, $\{u^i_2, u^i_4\}$, $\{u^i_3, u^i_4\}$, $\{u^i_3, p^i_1\}$ and $\{u^i_4, p^i_j\}$ for $j \in \{1, \dots, 4\}$.
	\end{enumerate}
	
	\begin{restatable}{lemma}{pathlemvariable}\label{path:lem:variable}
		The online algorithm has to traverse the variable gadget for $x_i$ from $u^i_1$ to $u^i_7$ and uses either every vertex $v^i_j$ or every vertex $w^i_j$, and none of the other, respectively. 
		Furthermore, it cannot leave a variable gadget by entering one of the clauses early.
	\end{restatable}
	% !TEX root=lipics-paper.tex
\begin{proof}
	We actually prove the slightly stronger statement that the online algorithm first has to visit the vertices $u^i_2$ and $u^i_3$ (in any order), then $u^i_4$, then the vertices $p^i_1$ to $p^i_4$, then $u^i_5$, then either the vertices $v^i$ and all $v^i_j$ or the vertices $w^i$ and all $w^i_j$, and finally the vertices $u^i_6$ and $u^i_7$.
	For a contradiction, assume it uses any other path.
	We consider the gadget part by part.
	\begin{description}
		\item[Part 4.] Assume the algorithm skips either $u^i_2$, $u^i_3$ or both $u^i_2$ and $u^i_4$.
		Then, it has to choose a new vertex as the next vertex when it is at $u^i_2$ or $u^i_3$.
		However, these vertices both have a sink attached, which is indistinguishable from the other new vertices.
		Thus, the adversary presents the vertices such that the online algorithm chooses a sink as its next vertex.
		Since there is no way to leave the sink, this prevents the algorithm from finding an $s$-$t$ path.
		On the other hand, if the algorithm always chooses a known neighbor as its next vertex after leaving $u^i_1$, it is able to find its way to $u^i_4$, as described above.
		\item[Part 3.] Assume the algorithm skips parts of the path from $p^i_1$ to $p^i_4$.
		Consequently, it either enters this path via an edge other than the one from $u^i_4$ to $p^i_1$, or it leaves this path through any edge that is not the one from $p^i_4$ to $u^i_5$.\\
		In the first case, it has to choose a new neighbor of $u^i_4$, but since $u^i_4$ has a sink attached, it is forced to enter a sink and thus loses the game by the same arguments as above.
		Therefore, the algorithm has to use the edge $\{u^i_4, p^i_1\}$.\\
		In the second case, it again has to choose a new neighbor of any of the vertices $p^i_j$, since none of them have a common neighbor, except $u^i_4$, which is already visited, and $u^i_5$, which only has edges to $p^i_1$ and $p^i_4$.
		Thus, it loses the game due to the sinks attached to $p^i_j$.
		However, if the online algorithm always chooses a known neighbor as its next vertex from $u^i_4$ to $u^i_5$, it traverses every vertex $p^i_j$ and ends up at $u^i_5$.
		\item[Part 2.] The only choice the algorithm has in this part of the gadget is $v^i$ or $w^i$, and thus it cannot deviate from the path described above.
		\item[Part 1.] Assume the algorithm leaves the variable gadget early through a clause gadget.
		Then, it has to choose a new neighbor of $v^i_2$ or $w^i_4$.
		Both of these vertices have a sink attached, and thus the adversary can always force a loss.
		However, if the online algorithm always prefers a known neighbor, it ends up in $u^i_6$.
		From $u^i_6$, it has to choose the known neighbor $u^i_7$, since the vertex $v^i_5$ (resp. $w^i_5$) is indistinguishable from the sink attached to $u^i_6$.
	\end{description}
\end{proof}
	
	With the previous lemma, we proved that the online algorithm always has to traverse a variable gadget in a way that assigns it either true or false.
	It remains to show that the quantifiers are correctly simulated.
	
	\begin{restatable}{lemma}{pathlemvariableChoice}\label{path:lem:variableChoice}
		When simulating the TQBF game with the reduction, the following holds. 
		If a variable $x_i$ is $\exists$-quantified, the online algorithm is able to choose its truth assignment.
		On the other hand, if $x_i$ is $\forall$-quantified, the adversary is able to choose its truth assignment, and the online algorithm learns that truth assignment before choosing the next $\exists$-variable.
	\end{restatable}
	% !TEX root=lipics-paper.tex
\begin{proof}
	If the variable $x_i$ is $\exists$-quantified, the variable gadget contains the edge $\{p^i_2, v^i\}$. 
	Due to \Cref{path:lem:variable}, the online algorithm always visits $p^i_2$ before $u^i_5$.
	Therefore, it can distinguish $v^i$ and $w^i$, as $v^i$ is known and $w^i$ is new, and choose the truth assignment of variable $x_i$. \\
	If the variable $x_i$ is $\forall$-quantified, that edge does not exist.
	Therefore, both $v^i$ and $w^i$ are new neighbors when the online algorithm is on $u^i_5$.
	Since they are indistinguishable, the adversary can choose the ordering in which it presents them such that the online algorithm chooses the vertex corresponding to the truth assignment the adversary wants. \\
	Once the algorithm traverses the path $v^i_j$ (resp. $w^i_j$), it learns the truth assignment from the first vertex that has at least one new neighbor and a known one:
	If $j = 2$, the path corresponds to true and if $j = 4$, the path corresponds to false.
	Therefore, it can deduce the truth assignment it chose for $x_i$ before it has to choose the truth assignment of the next $\exists$-quantified variable.
\end{proof}

\paragraph*{Clause Gadget}
\label{path:def:clauseGadget}
	The clause gadget for clause $C_i$ consists of two disconnected vertices, connected by one path for every variable they contain.
	Then, the clause gadget, also shown in \Cref{st:fig:clause}, is defined as follows:
	\begin{itemize}
		\item There are two vertices $c^i_1$ and $c^i_2$, which are not connected and each have a sink attached.
		\item For each literal $\ell \in C_i$, there are two vertices: $c^i_\ell, c^i_{\ell'}$.
		\item For each literal $\ell \in C_i$, the vertices $c^i_1, c^i_\ell, c^i_{\ell'}$ and $c^i_2$ form a path.
		\item For each literal $\ell \in C_i$, let $var(\ell)$ be the index of its corresponding variable. If $\ell$ is non-negated, there is the edge $\{c^i_\ell, v^{var(\ell)}_2\}$, and if $\ell$ is negated, there is the edge $\{c^i_\ell, w^{var(\ell)}_4\}$.
	\end{itemize}
	% !TEX root=lipics-paper.tex
\begin{figure}[t]
	% colors
	\def\sinkcolor{red}
	\def\revealcolorid{blue}\def\revealcoloriddot{dashed}
	\def\wincolor{dkgreen}\def\wincolordot{densely dash dot}
	\def\revealcolortfe{dkgreen}\def\revealcolortfedot{densely dash dot}
	\def\revealcolorva{cyan}\def\revealcolorvadot{densely dash dot dot}
	\def\revealcolorclause{mauve}\def\revealcolorclausedot{dotted}
	\def\clausecolor{orange}\def\clausecolordot{densely dotted}
		
	\def\j{0}
	\def\k{10}
	\def\l{20}

	\centering
	\resizebox{0.73\textwidth}{!}{
		\begin{tikzpicture}[scale=0.4,
			node/.style = {shape=circle, draw, inner sep=0pt, minimum size=0.25cm},
			textnode/.style = {shape=circle, draw, inner sep=0pt, minimum size=0.4cm},
			sink/.style = {shape=circle, draw=\sinkcolor, inner sep=0pt, minimum size=0.1cm},
			box/.style = {rectangle, fill=gray!20, rounded corners, fill opacity=1}]
			% --------------------------------------------------------------------------------
			% box
			% --------------------------------------------------------------------------------
			\node (copysc1s) at (\k+1, -1) {};
			\node (copyc1s\j) at (\j+2, 3) {};
			\node (copyc1t\l) at (\l+4, 3) {};
			\node[box, fit=(copysc1s)(copyc1s\j)(copyc1t\l)] (box) {};
			
			% --------------------------------------------------------------------------------
			% clause
			% --------------------------------------------------------------------------------
			\node[textnode] (c1s) at (\k+1, 0) {$c^i_1$};
			\node[textnode] (c1t) at (\k+5, 0) {$c^i_2$};
			\node (c1) at (\k+3, -0.6) {\huge$C_i$};
			% sinks
			\node[sink] (sc1s) at (\k+1, -1.2) {};
			\node[sink] (sc1t) at (\k+5, -1.2) {};
			\path[-, draw=\sinkcolor] (c1s) edge (sc1s) (c1t) edge (sc1t);
			% paths that get revealed by variables
			\foreach \x in {\j,\k,\l} {
				\node[node] (c1s\x) at (\x+2, 3) {};
				\node[node] (c1t\x) at (\x+4, 3) {};
				\path[-] (c1s) edge (c1s\x) (c1s\x) edge (c1t\x) (c1t\x) edge (c1t);
			}
			% in and outgoing edges
			\node[textnode] (s) at (\j-2, 0) {$s$};
			\node[textnode] (t) at (\l+8, 0) {$t$};
			\path[-, dashed] (s) edge (c1s) (c1t) edge (t);
			
			% --------------------------------------------------------------------------------
			% variables
			% --------------------------------------------------------------------------------
			\foreach \x/\n in {\j/j,\k/k,\l/\ell} {
				% name
				\node at (\x, 6.5) {\huge$x_\n$};
				% part of variable gadget
				\node[node] (xt\x0) at (\x, 5) {};
				\node[node] (xt\x1) at (\x+2, 5) {};
				\node[node] (xt\x2) at (\x+4, 5) {};
				\node[node] (xt\x3) at (\x+6, 5) {};
				\path[-] (xt\x0) edge (xt\x1) (xt\x1) edge (xt\x2) (xt\x2) edge (xt\x3);
				% sink
				\node[sink] (sxt\x1) at (\x+2, 6) {};
				\path[-, draw=\sinkcolor] (xt\x1) edge (sxt\x1);
				% revealing edges
				\node[inner sep=0pt] (xt\x2h) at (\x+4, 7) {};
				\path[-, draw=\revealcolorva, \revealcolorvadot, thick] (xt\x2) edge (xt\x2h);
				\path[-, draw=\clausecolor, \clausecolordot, thick] (xt\x1) edge (c1s\x);
			}
			
		\end{tikzpicture}
	}
	\caption{Clause gadget for the reduction from \qsatgame{} to \ostpg. The dashed lines indicate that there might be more gadgets in between. Only parts of the variable gadgets are shown.}
	\label{st:fig:clause}
\end{figure}
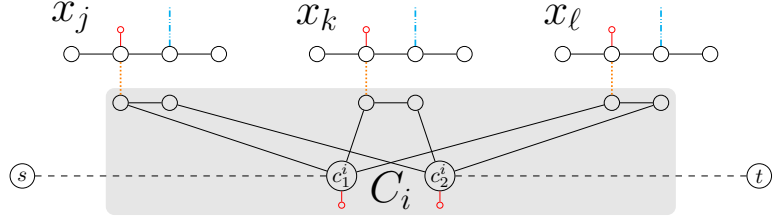
	
	\begin{restatable}{lemma}{pathlemclause}\label{path:lem:clause}
		If at least one variable satisfies $C_i$ in the assignment chosen by the online algorithm, it can find a path from $c^i_1$ to $c^i_2$.
		Otherwise, the online algorithm loses the game.
		Furthermore, it cannot use the variable gadgets to enter a not yet visited clause gadget.
	\end{restatable}
	% !TEX root=lipics-paper.tex
\begin{proof}
	Let $\ell_1$ be a literal that satisfies $C_i$ in the assignment chosen by the online algorithm.
	Then, the vertex $c^i_{\ell_1}$ has been revealed while traversing the variable gadget corresponding to $\ell_1$.
	On the other hand, let $\ell_2$ be a literal that does not satisfy $C_i$ in the assignment chosen by the online algorithm.
	Then, the vertex $c^i_{\ell_2}$ has not been revealed while traversing the variable gadget corresponding to $\ell_2$ and is thus new. 
	
	Therefore if the online algorithm chooses any known neighbor of $c^i_1$, it uses a path that leads to $c^i_2$, satisfying the clause $C_i$.
	W.l.o.g. we may assume that the online algorithm assigned $\ell$ the value true.
	Then the vertex $v^{var(\ell)}_2$ has already been visited, and cannot be visited again to reach a different clause.
	Further, when at $c^i_2$, the algorithm cannot use any of the other paths connecting $c^i_1$ and $c^i_2$, since $c^i_2$ has a sink attached and the vertices $c^i_{\ell'_a}$ for any other $\ell_a \in C_i$ are new and thus indistinguishable from the sink. 
	
	If no variable satisfies $C_i$ in the assignment chosen by the online algorithm, all neighbors of $c^i_1$ are new by the arguments above.
	Thus, the adversary can force a loss since $c^i_1$ has a sink attached.
\end{proof}

\paragraph*{The Complete Reduction}
	Given a \qsatgame{} instance with formula $\varphi$ containing variables $X$ and clauses $C$ we create an instance of \ostupg{} as follows:
	For each variable and each clause, a gadget is created as described above, and the clause gadgets are connected to the variable gadgets depending on the variables they contain.
	Two distinct vertices $s$ and $t$ are created, and the algorithm has to find a path from $s$ to $t$.
	The edge $\{s, u^1_1\}$ as well as the edges $\{u^i_7, u^{i+1}_1\}$ for $i \in \{1, \dots, |X|-1\}$ are added.
	Further, the edges $\{u^{|X|}_7, c^1_1\}$, $\{c^i_2, c^{i+1}_1\}$ for $i \in \{1, \dots, |C|-1\}$ and $\{c^{|C|}_2, t\}$ are introduced.
	Finally, for $i \in \{1, \dots, |C|-1\}$, the edge $\{c^i_1, c^{i+1}_1\}$ is added (purple and loosely dotted in \Cref{st:fig:complete}) as well as the edge $\{c^{|C|}_1, t\}$.
	This is necessary, since otherwise the next vertex the algorithm is supposed to choose when at $c^i_2$, is indistinguishable from the sink attached to it.
	An example of this construction can be seen in \Cref{st:fig:complete}.
	With this, we can now prove \Cref{path:thm:uPathPspaceHard}.
	% !TEX root=lipics-paper.tex
\begin{figure}[t]
	% colors
	\def\sinkcolor{red}
	\def\revealcolorid{blue}\def\revealcoloriddot{dashed}
	\def\wincolor{dkgreen}\def\wincolordot{densely dash dot}
	\def\revealcolortfe{dkgreen}\def\revealcolortfedot{densely dash dot}
	\def\revealcolorva{cyan}\def\revealcolorvadot{densely dash dot dot}
	\def\revealcolorclause{mauve}\def\revealcolorclausedot{dotted}
	\def\clausecolor{orange}\def\clausecolordot{densely dotted}
	\def\winpath{1}
	% positions of gadgets
	\def\ione{0}
	\def\itwo{9}
	\def\ithree{18}
	\def\jone{3}
	\def\jtwo{11}
	\def\jthree{19}
	\def\rowtwo{-6}
	\centering
	\resizebox{\textwidth}{!}{
		\begin{tikzpicture}[scale=0.375,
			node/.style = {shape=circle, draw, inner sep=0pt, minimum size=0.25cm},
			textnode/.style = {shape=circle, draw, inner sep=0pt, minimum size=0.4cm},
			smallnode/.style = {shape=circle, draw, inner sep=0pt, minimum size=0.1cm},
			sink/.style = {shape=circle, draw=\sinkcolor, inner sep=0pt, minimum size=0.1cm},
			box/.style = {rectangle, fill=gray!20, rounded corners, fill opacity=1, inner sep=1pt}]
			% --------------------------------------------------------------------------------
			% boxes
			% --------------------------------------------------------------------------------
			% variables
			\foreach \x/\y in {1/\ione, 2/\itwo, 3/\ithree} {
				\node (ulx\x) at (\y, 2) {};
				\node (brx\x) at (\y+8, -2) {};
				\node[box, fit=(ulx\x)(brx\x)] (boxx\x) {};
				\node[below right, inner sep=3pt] (textx\x) at (boxx\x.north west) {$x_\x$};
			}
			% clauses
			\foreach \x/\y in {1/\jone, 2/\jtwo, 3/\jthree} {
				\node (ulc\x) at (\y-1.5, \rowtwo+1) {};
				\node (brc\x) at (\y+4.5, \rowtwo-0.5) {};
				\node[box, fit=(ulc\x)(brc\x)] (boxc\x) {};
				\node[below right, inner sep=3pt] (textc\x) at (boxc\x.north west) {$C_\x$};
			}
			
			% --------------------------------------------------------------------------------
			% variables
			% --------------------------------------------------------------------------------
			\foreach \a in {\ione, \itwo, \ithree} {
				% ID gadget of ID gadget 
				\node[smallnode] (id11\a) at (\a, 0) {};
				\node[smallnode] (id12\a) at (\a+0.5, 0.5) {};
				\node[smallnode] (id13\a) at (\a+0.5, -0.5) {};
				\path[-] (id11\a) edge (id12\a) edge (id13\a);
				\path[-] (id12\a) edge (id13\a);
				% sinks
				\node[sink] (sid12\a) at (\a+0.5, 1) {};
				\node[sink] (sid13\a) at (\a+0.5, -1) {};
				\path[-, draw=\sinkcolor] (id12\a) edge (sid12\a);
				\path[-, draw=\sinkcolor] (id13\a) edge (sid13\a);
				
				% ID gadget
				\node[smallnode] (id20\a) at (\a+1, 0) {};
				\node[smallnode] (id21\a) at (\a+2.5, 2) {};
				\path[-] (id20\a) edge (id21\a);
				\foreach \x in {2,3,4} {
					\pgfmathtruncatemacro{\y}{\x-1}
					\node[smallnode] (id2\x\a) at (\a+2.5, -1.33*\x+3.33) {};
					\path[-] (id2\y\a) edge (id2\x\a);
				}
				% extra edges that reveal the path
				\foreach \x in {2,3,4} {
					\path[-, draw=\revealcolorid, \revealcoloriddot] (id20\a) edge (id2\x\a);
				}
				% sink at start node
				\node[sink] (sid20\a) at (\a+1, -0.5) {};
				\path[-, draw=\sinkcolor] (id20\a) edge (sid20\a);
				% sinks along the path
				\foreach \x in {1,...,4} {
					\node[sink] (sid2\x\a) at (\a+2, -1.33*\x+3.33) {};
					\path[-, draw=\sinkcolor] (id2\x\a) edge (sid2\x\a);
				}
				% connections from IDID gadget
				\path[-] (id12\a) edge (id20\a);
				\path[-] (id13\a) edge (id20\a);
				\path[-, draw=\revealcolorid, \revealcoloriddot] (id12\a) edge (id21\a);
				
				% decision and revealing part
				\node[smallnode] (xd\a) at (\a+3.5, 0) {};
				\node[smallnode] (xt\a) at (\a+4, 1) {};
				\node[smallnode] (xt0\a) at (\a+4.5, 2) {};
				\node[smallnode] (xf\a) at (\a+4, -1) {};
				\node[smallnode] (xf0\a) at (\a+4.5, -2) {};
				\path[-] (xd\a) edge (xt\a) (xt\a) edge (xt0\a);
				\path[-] (xd\a) edge (xf\a) (xf\a) edge (xf0\a);
				% paths corresponding to true and false
				\foreach \x in {1,...,4} {
					\pgfmathtruncatemacro{\y}{\x-1}
					\node[smallnode] (xt\x\a) at (\a+4.5+0.5*\x, 2) {};
					\path[-] (xt\y\a) edge (xt\x\a);
				}
				\foreach \x in {1,...,4} {
					\pgfmathtruncatemacro{\y}{\x-1}
					\node[smallnode] (xf\x\a) at (\a+4.5+0.5*\x, -2) {};
					\path[-] (xf\y\a) edge (xf\x\a);
				}
				% end of gadget
				\node[smallnode] (xe1\a) at (\a+7.5, 0) {};
				\node[smallnode] (xe2\a) at (\a+8, 0) {};
				\path[-] (xe1\a) edge (xt4\a) edge (xf4\a) edge (xe2\a);
				% sinks 
				\node[sink] (sxt1\a) at (\a+5, 1.5) {};
				\path[-, draw=\sinkcolor] (xt1\a) edge (sxt1\a);
				\node[sink] (sxf3\a) at (\a+6, -1.5) {};
				\path[-, draw=\sinkcolor] (xf3\a) edge (sxf3\a);
				\node[sink] (sxe1\a) at (\a+7, 0) {};
				\path[-, draw=\sinkcolor] (xe1\a) edge (sxe1\a);
				% edges connecting ID gadget to variable decision
				\path[-, draw=\revealcolorid, \revealcoloriddot] (id21\a) edge (xd\a);
				\path[-] (id24\a) edge (xd\a);
				
				% reveal edges
				\ifthenelse{\a = \itwo}{}{\path[-, draw=\revealcolortfe, \revealcolortfedot] (id22\a) edge (xt\a);}
				\path[-, draw=\revealcolorva, \revealcolorvadot, out=50, in=250, out distance=3cm, in distance=2cm] (id23\a) edge (xt2\a);
				\path[-, draw=\revealcolorva, \revealcolorvadot, out=355, in=110, out distance=4cm, in distance=1cm] (id23\a) edge (xf4\a);
				\path[-, draw=\revealcolorva, \revealcolorvadot, out=0, in=225, out distance=3cm, in distance=1cm] (id23\a) edge (xe2\a);
			}
			
			% --------------------------------------------------------------------------------
			% clauses
			% --------------------------------------------------------------------------------
			\foreach \a in {\jone, \jtwo, \jthree} {
				\node[smallnode] (cs\a) at (\a+1, \rowtwo) {};
				\node[smallnode] (ct\a) at (\a+3, \rowtwo) {};
				% sinks
				\node[sink] (scs\a) at (\a+1, \rowtwo-0.5) {};
				\node[sink] (sct\a) at (\a+3, \rowtwo-0.5) {};
				\path[-, draw=\sinkcolor] (cs\a) edge (scs\a) (ct\a) edge (sct\a);
				% paths that get revealed by variables
				\foreach \x in {0,2,4} {
					\node[smallnode] (cs\x\a) at (\a+\x-0.5, \rowtwo+1) {};
					\node[smallnode] (ct\x\a) at (\a+\x+0.5, \rowtwo+1) {};
					\path[-] (cs\a) edge (cs\x\a) (cs\x\a) edge (ct\x\a) (ct\x\a) edge (ct\a);
				}
			}
			
			% --------------------------------------------------------------------------------
			% clause variable connections
			% --------------------------------------------------------------------------------
			\path[-, draw=\clausecolor, \clausecolordot, out=90, in=240, out distance=1cm, in distance=1cm] (cs0\jone) edge (xt1\ione);
			\path[-, draw=\clausecolor, \clausecolordot, out=90, in=250, out distance=2cm, in distance=5cm] (cs2\jone) edge (xt1\itwo);
			\path[-, draw=\clausecolor, \clausecolordot, out=90, in=250, out distance=2cm, in distance=7cm] (cs4\jone) edge (xt1\ithree);
			
			\path[-, draw=\clausecolor, \clausecolordot, out=90, in=300, out distance=1cm, in distance=3cm] (cs0\jtwo) edge (xt1\ione);
			\path[-, draw=\clausecolor, \clausecolordot, out=90, in=270, out distance=1cm, in distance=1cm] (cs2\jtwo) edge (xf3\itwo);
			\path[-, draw=\clausecolor, \clausecolordot, out=90, in=270, out distance=1cm, in distance=1cm] (cs4\jtwo) edge (xf3\ithree);
			
			\path[-, draw=\clausecolor, \clausecolordot, out=90, in=270, out distance=1cm, in distance=1cm] (cs0\jthree) edge (xf3\ione);
			\path[-, draw=\clausecolor, \clausecolordot, out=90, in=300, out distance=1cm, in distance=3cm] (cs2\jthree) edge (xt1\itwo);
			\path[-, draw=\clausecolor, \clausecolordot, out=100, in=240, out distance=2cm, in distance=1.5cm] (cs4\jthree) edge (xt1\ithree);
			
			% --------------------------------------------------------------------------------
			% remaining connections
			% --------------------------------------------------------------------------------
			\node[textnode] (s) at (\ione-2, 0) {$s$};
			\node[textnode] (t) at (\jthree+6, \rowtwo) {$t$};
			\node[smallnode] (help1) at (\ithree+8, -3) {};
			\node[smallnode] (help2) at (\jone-3, \rowtwo+2) {};
			% paths
			\path[-] (s) edge (id11\ione) (xe2\ione) edge (id11\itwo) (xe2\itwo) edge (id11\ithree) (xe2\ithree) edge (help1);
			\path[-, out=180, in=0, out distance=0.1cm, in distance=0.1cm] (help1) edge (help2);
			\path[-, out=270, in=180, out distance=1.5cm, in distance=2.5cm] (help2) edge (cs\jone);
			\path[-] (ct\jone) edge (cs\jtwo) (ct\jtwo) edge (cs\jthree) (ct\jthree) edge (t);
			\path[-, draw=\revealcolorclause, \revealcolorclausedot, out=300, in=240, out distance=1.2cm, in distance=1.2cm] (cs\jone) edge (cs\jtwo);
			\path[-, draw=\revealcolorclause, \revealcolorclausedot, out=300, in=240, out distance=1.2cm, in distance=1.2cm] (cs\jtwo) edge (cs\jthree);
			\path[-, draw=\revealcolorclause, \revealcolorclausedot, out=300, in=240, out distance=1cm, in distance=0.8cm] (cs\jthree) edge (t);
			
			% --------------------------------------------------------------------------------
			% winning path
			% --------------------------------------------------------------------------------
			
		\end{tikzpicture}
	}
	\caption{
		Sketch of the full construction of the reduction from \qsatgame{} to \ostupg{} with the formula $\exists x_1 \forall x_2 \exists x_3 (x_1 \vee x_2 \vee x_3) \wedge (x_1 \vee \overline{x}_2 \vee \overline{x}_3) \wedge (\overline{x}_1 \vee x_2 \vee x_3)$.
	}
	\label{st:fig:complete}
\end{figure}
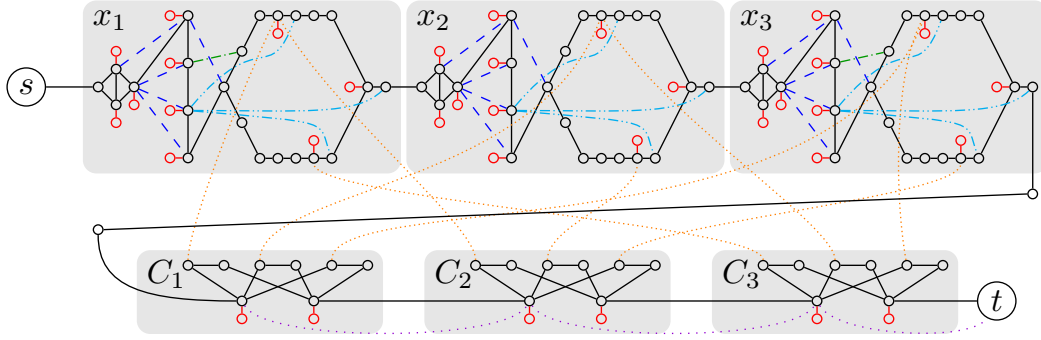
	% !TEX root=lipics-paper.tex
\begin{proof}[Proof of \Cref{path:thm:uPathPspaceHard}]
	\Cref{path:lem:clause} still holds when attaching the additional edges to all clauses and also no clause gadget can be skipped, since $c^i_1$ has a sink attached and $c^{i+1}_1$ is a new vertex when at $c^i_1$.
	Furthermore, the neighbor $c^{i+1}_1$ (resp. $t$ for $i = |C|$) of $c^i_2$ is distinguishable from any other neighbor of $c^i_2$, because it is the only known one.
	Thus by \Cref{path:lem:variable,path:lem:variableChoice,path:lem:clause}, the online algorithm can find a path from $s$ to $t$ in the graph described above if and only if the \qsatgame{} instance has a winning strategy for the $\exists$-player.
	All our gadgets have constant size.
	Therefore, our reduction runs in polynomial time, and the claim follows.
\end{proof}

	\paragraph*{Directed Graphs}
	Next, we show that the above ideas can also be applied to directed graphs.
	The \ostdpg{} is defined analogously to the \ostupg{}.
	
	\begin{restatable}{corollary}{pathcordPathPspaceHard}\label{path:cor:dPathPspaceHard}
		The \ostdpg{} is \PSPACE-complete.
	\end{restatable}
	% !TEX root=lipics-paper.tex
\begin{proof}
	We use a similar construction as for \Cref{path:thm:uPathPspaceHard}, but replace the undirected edges by directed arcs as follows:
	\begin{itemize}
		\item Any edge connecting a vertex with a sink is replaced with an arc directed towards the sink.
		\item Any edge, that only reveals vertices but cannot be used by the arguments of \Cref{path:lem:variable,path:lem:clause,path:thm:uPathPspaceHard}, is replaced with an arc directed towards the revealed vertex.
		\item Any edge excluding $\{u^i_2, u^i_3\}$, for $i \in \{1, \dots, |X|\}$, can only be traversed in one direction (by the arguments of \Cref{path:lem:variable,path:lem:clause,path:thm:uPathPspaceHard}) and is therefore replaced by the respective directed arc.
		\item The edges $\{u^i_2, u^i_3\}$, for $i \in \{1, \dots, |X|\}$, are replaced by arcs $(u^i_2, u^i_3)$ and $(u^i_3, u^i_2)$, for $i \in \{1, \dots, |X|\}$.
	\end{itemize}
	Additionally, the arcs $(c^i_1, c^{i+1}_1)$, for $i \in \{1,\dots,|C|-1\}$, and $(c^{|C|}_1, t)$ as well as the sinks attached to $c^i_2$ are removed.
	This leaves $c^i_2$ with only one outgoing arc and thus the algorithm can still find its way from $c^i_2$ to $c^{i+1}_1$.
	Then, the claim follows by the same arguments as in \Cref{path:lem:variable,path:lem:variableChoice,path:lem:clause,path:thm:uPathPspaceHard}.
\end{proof}
 
\subsection{Trails}
\label{sec:trail}
In this section, we consider online games based on \stt.
Like for paths, we consider directed and undirected graphs separately.
The definitions for \ostutg{} and \ostdtg{} are analogous to \Cref{path:def:ostupg}, replacing paths with trails.
We show \PSPACE-hardness for both problems.

For the \ostutg{}, we use a very similar reduction as for \Cref{path:thm:uPathPspaceHard}.
In the variable gadget, we replace every vertex $p^i_j$ with two vertices $p^i_j$ and $q^i_j$, connected by an edge $\{p^i_j, q^i_j\}$, with the sinks attached to $q^i_j$ instead of $p^i_j$.
The edges $\{u^i_3, p^i_1\}$, $\{u^i_4, p^i_j\}$ remain unchanged, for $j \in \{1, \dots, 4\}$.
The edges $\{p^i_j, p^i_{j+1}\}$ are replaced by $\{q^i_j, p^i_{j+1}\}$, for $j \in \{1, \dots, 3\}$.
All remaining edges starting in $p^i_j$ start in $q^i_j$ instead, for $j \in \{1, \dots, 4\}$. 
This change prevents the algorithm from visiting $u^i_4$ multiple times.

\begin{restatable}{lemma}{pathlemuTrailPspaceHard}\label{path:lem:uTrailPspaceHard}
	\ostutg{} is \PSPACE-complete.
\end{restatable}
% !TEX root=lipics-paper.tex
\begin{proof}
	First, we show that \Cref{path:lem:variable} also applies to the \ostutg{} with this slight modification.
	We note that the arguments for the \ostupg{} do not change with the modification, except that the path $p^i_1$ to $p^i_4$ is replaced by the path from $p^i_1$ to $q^i_4$.
	In the following, we only argue about additional cases that arise from requiring the online algorithm to find a trail instead of a path.
	Note that, sinks still work as before, since leaving them requires using an edge twice.
	Again, we consider the variable gadget part by part.
	\begin{description}
		\item[Part 4.] Assume the algorithm returns to $u^i_4$ from $u^i_3$ or $u^i_2$, depending on which edge is still unused. 
		When at $u^i_2$, the only valid next vertex is its sink, while in $u^i_3$ the only valid moves are either the sink or $p^i_1$.
		Thus, the algorithm still either loses or ends up at $p^i_1$.
		\item[Part 3.] If the algorithm returns to Part 4 of the gadget, the adversary prevents it from reaching Part 3 again:
		Assume the algorithm returns to $u^i_3$ from $p^i_1$.
		Then, the adversary forces it to choose the sink attached to $u^i_3$ or $u^i_4$ (if that edge has not been used), since the vertices $p^i_j$ for $j > 1$ are still new neighbors of $u^i_4$.
		Now assume the algorithm returns to $u^i_4$ from $p^i_j$, for $j > 1$.
		Then, it loses the game for the same reason as above. 
		
		Otherwise if the algorithm always chooses the new neighbor when at $p^i_j$, and a known neighbor when at $q^i_j$, it ends up at $u^i_5$.
		\item[Part 2.] The only difference occurs if $x_i$ is $\exists$-quantified, since then the algorithm can use the edge $\{q^i_2, v^i\}$ to return to $q^i_2$.
		Since the only valid next move is the sink attached to $q^i_2$, the algorithm loses the game.
		\item[Part 1.] The algorithm can return to $q^i_3$ from either $v^i_2$, $w^i_4$ or $u^i_7$.
		If the algorithm returns from $v^i_2$ or $w^i_4$, all unused edges from $q^i_3$ lead to vertices that are only known from the first visit of $q^i_3$, allowing the adversary to force a loss because one of them is a sink.
		On the other hand if the algorithm returns from $u^i_7$, additionally one of the unused edges leads to the already visited vertex $v^i_2$ or $w^i_4$.
		Choosing that vertex also results in a loss, because all adjacent edges are already used.
	\end{description}
	\Cref{path:lem:variableChoice} and its proof still apply to the \ostutg{} with this modification.
	For \Cref{path:lem:clause}, we need to consider the additional case that the algorithm chooses to return to a vertex $v^j_2$ (resp. $w^j_4$) from the clause gadget of a clause $C_i$. 
	From there it can (if there are any) visit a vertex $c^{i'}_j$ it has already visited while traversing the clause gadget of clause $C_{i'}$.
	Then, it is stuck in that vertex, as all its adjacent edges have been used.
	The algorithm can also choose to visit a neighbor of $v^j_2$ (resp. $w^j_4$) that is only known from the first visit of $v^j_2$ (resp. $w^j_4$).
	Then, the algorithm loses the game, since one of them is a sink.
	Finally if we consider the connections between clause gadgets, the algorithm can return to a previous clause by using the edge $\{c^i_1, c^{i+1}_1\}$.
	However, that results in it getting stuck as it exhausted all identifiable edges leading to $c^{i+1}_1$.
	Thus by the same arguments as in \Cref{path:thm:uPathPspaceHard}, the claim follows.
\end{proof}

For the directed case, we use the same construction as in \Cref{path:cor:dPathPspaceHard}.
	
\begin{restatable}{corollary}{pathcordTrailPspaceHard}\label{path:cor:dTrailPspaceHard}
	\ostdtg{} is \PSPACE-complete.
\end{restatable}
% !TEX root=lipics-paper.tex
\begin{proof}
	This directly follows from the arguments of \Cref{path:cor:dPathPspaceHard}, since the direction of the arcs forces any trail to be a path, because it is not possible to return to previous vertices.
	The exception to this are the vertices $u^i_1, u^i_2, u^i_3, u^i_4$, where the algorithm can choose to return to $u^i_2$ from $u^i_3$ or vice versa, however, it is still forced to traverse all these vertices, and thus $u^i_4$ last.
\end{proof}
	
\subsection{Walks}
\label{sec:walk}
In this section, we consider online games based on \stw. 
Like for paths, we consider directed and undirected graphs separately.
The definitions for \ostuwg{} and \ostdwg{} are analogous to \Cref{path:def:ostupg}, replacing paths with walks.

First, we show that the problem for walks on undirected graphs is much easier than for trails or paths.

\begin{restatable}{lemma}{pathlemuwgLComplete}\label{path:lem:uwgLComplete}
	\ostuwg{} is \LSPACE-complete.	
\end{restatable}
% !TEX root=lipics-paper.tex
\begin{proof}
    If the graph is connected, then there is a strategy for the online algorithm.
    For example, it can always find an $s$-$t$ walk by just performing a depth first search, as that always leads to discovering $t$, since $G$ is finite.
    Thus the problem of deciding whether there is a strategy for the online algorithm is as complex as deciding whether $s$ and $t$ are in the same connected component.
    The undirected s-t connectivity problem was proven to be \LSPACE-complete in \cite{DBLP:journals/jacm/Reingold08}.
    Therefore, \ostuwg{} is also \LSPACE-complete.
\end{proof}
	
However for the \ostdwg{}, we can use the same reduction as in \Cref{path:cor:dPathPspaceHard}.
	
\begin{restatable}{corollary}{pathcordWalkPspaceHard}\label{path:cor:dWalkPspaceHard}
	\ostdwg{} is \PSPACE-complete for general graphs.
\end{restatable}
% !TEX root=lipics-paper.tex
\begin{proof}
	This follows from the same argument as \Cref{path:cor:dTrailPspaceHard}.
\end{proof}
	
\begin{restatable}{lemma}{pathlemdwgScLComplete}\label{path:lem:dwgScLComplete}
	Strongly connected graphs are trivial yes-instances of \ostdwg{}.
\end{restatable}
% !TEX root=lipics-paper.tex
\begin{proof}
	The statement follows from the same argument as in the proof of \Cref{path:lem:uwgLComplete}.
\end{proof}
	
This completes our analysis of online games based on \stp, \stt{} and \stw, when the graph has no costs on its edges and the solution may be of any length.
For \stp{} and \stt{}, introducing edge costs and asking for a path (resp. trail) of a specific cost obviously still results in \PSPACE-hardness.
However for the online games based on \stw{}, we show in the following that even unit edge costs make the problem \PSPACE-hard, both in the directed and undirected case.

\begin{definition}[\ostuwwg]\label{path:def:ostuwwg}\hfill\\
	\textbf{Given:} An undirected graph $G = (V, E, cost)$, two vertices $s, t \in V$ and a number $k \in \mathbb{N}$. \\
	\textbf{Question:} Does an online algorithm exist, that finds a walk from $s$ to $t$ in $G$ of total cost at most $k$ for all strategies of the adversary in the neighborhood reveal model, while knowing a map of $G$?
\end{definition}

The \ostdwwg{} is defined analogously.
We first handle the easy case of \ostdwwg.
For that, we slightly modify the reduction from \Cref{path:cor:dPathPspaceHard}.
	
\begin{restatable}{lemma}{pathlemdwwgPspaceComplete}\label{path:lem:dwwgPspaceComplete}
	\ostdwwg{} is \PSPACE-complete with unit arc costs even if the graph is strongly connected.
\end{restatable}
% !TEX root=lipics-paper.tex
\begin{proof}
	If we assign every arc a cost of $1$, the path described in \Cref{path:thm:uPathPspaceHard} has a cost of $17n + 4m + 1$, where $n$ is the number of variables and $m$ is the number of clauses.
	Further, this is also the shortest path the algorithm can find, as argued in \Cref{path:cor:dPathPspaceHard}.
	On the other hand, if the formula is not satisfied, the adversary can force the algorithm to visit a sink due to \Cref{path:lem:clause}.
	Due to the direction of the arcs, any directed walk in $G$ is actually a path, except for the possibility to visit $u^i_2$ and $u^i_3$ multiple times.
	Thus, any walk the algorithm can find in $G$ is at least as long as that shortest path.
	Therefore, we set $k = 17n + 4m + 1$.\\
	To create a strongly connected graph $G'$, we first introduce a vertex $t'$. 
	Each sink as well as $t$ has an arc to $t'$.
	Finally, we add a path of length $2k$ from $t'$ back to $s$.
	Thus if the algorithm chooses a sink at any point, it is forced to use the path of length $2k$ to return to $s$ and consequently lose the game.
	Therefore, the algorithm can find a walk of length $k$ in $G'$ if and only if there is a winning strategy for the $\exists$-player in the \qsatgame{} instance.
\end{proof}
	
Next, we show that \ostuwwg{} is \PSPACE-hard.
As a first step, we show this result when $cost(e)$ is bounded by a polynomial.

\begin{restatable}{lemma}{pathlemuwWalkPspaceHard}\label{path:lem:uwWalkPspaceHard}
	\ostuwwg{} is \PSPACE-complete with edge costs bounded by a polynomial.
\end{restatable}
% !TEX root=lipics-paper.tex
\begin{proof}
	Let $M_1$ and $M_2$ be two large numbers, whose exact values will be specified later.
	Let $n := |X|$ and $m := |C|$.
	We use the same reduction as for \Cref{path:thm:uPathPspaceHard}, but remove all sinks except the one attached at $c^i_1$, as they are no longer needed.
	Additionally we introduce the following edge costs:
	\begin{itemize}
		\item For all $i \in \{1,\dots,n\}$, the edges $\{u^i_3, p^i_1\}$, $\{u^i_4, p^i_j\}$ for $j \in \{2,3,4\}$, $\{p^i_1, u^i_5\}$, $\{p^i_2, v^i\}$ (if it exists), $\{p^i_3, v^i_3\}$, $\{p^i_3, w^i_5\}$ and $\{p^i_3, u^i_7\}$ have cost $M_1$.
		\item For all $i \in \{1,\dots,n\}$ and $j \in \{1,\dots,m\}$, the edges $\{v^i_2, c^j_i\}$ and $\{w^i_4, c^j_i\}$ have cost $M_1$, if they exist.
		\item For all $i \in \{1,\dots,m-1\}$, the edge $\{c^i_1, c^{i+1}_1\}$ as well as $\{c^m_1, t\}$ have cost $M_1$.
		\item For all $i \in \{1,\dots,n\}$, the edges $\{v^i, v^i_1\}, \{w^i, w^i_1\}, \{v^i_5, u^i_6\}$ and $\{w^i_5, u^i_6\}$ have cost $M_2$.
		\item All remaining edges have cost $1$.
	\end{itemize}
	The cost of any shortest path from $c^1_1$ to $t$ is $4m$.
	Furthermore if the variable assignment chosen by the algorithm satisfies the formula of the \qsatgame{} instance, the algorithm can find such a path, by the arguments of \Cref{path:lem:clause,path:thm:uPathPspaceHard}.
	On the other hand, if the assignment does not fulfill the formula, the adversary can force costs of at least $4m + 2$.
	For that, let $i$ be the index of a clause that is not satisfied.
	By the arguments of \Cref{path:lem:clause}, the vertex $c^i_1$ only has new neighbors.
	Therefore, the adversary can force the algorithm to visit the sink attached to it.
	When it returns to $c^i_1$, all edges with cost $1$ allow the algorithm to follow a path of four edges to $c^{i+1}_1$, each with cost $1$.
	In total, this induces additional costs of $2$.
	
	Now let $M_2 = 40(n+m)$ and thus $M_2$ is larger than the total cost of the path that the online algorithm can use from $c^1_1$ to $t$ if all clauses are satisfied.
	Then any shortest path from $u^1_1$ to $c^1_1$ (independent of the variable assignment) has a total cost of $n(14+2M_2) = 80nm + 80n^2 + 14n$.
	Due to \Cref{path:lem:variable,path:thm:uPathPspaceHard}, the online algorithm can find at least one such path. 
	
	We set $k = 80nm + 80n^2 + 14n + 4m + 1$ (since the algorithm also needs to traverse the edge $\{s, u^1_1\}$).
	Further let $M_1 = 10k$.
	Thus if the online algorithm uses any edge with cost $M_1$, it immediately loses the game.
	Therefore, it cannot enter the clause gadgets early, or skip (parts of) variable or clause gadgets.
	Further, using more than $2n$ edges of cost $M_2$ also immediately forces a loss, since $(2n+1)M_2 = 80nm + 80n^2 + 40m + 40n > k$.
	Since any shortest path from $u^1_1$ to $c^1_1$ uses exactly $2n$ distinct edges of cost $M_2$, using any such edge twice also loses the game.
	Thus, the algorithm cannot assign any variable both the values true and false.
	Therefore, the algorithm has to use a path as described in \Cref{path:thm:uPathPspaceHard}.
	If at least one clause is not satisfied by its chosen assignment, it has an additional cost of at least $2$ and loses the game.
	Overall, the algorithm can win the game if and only if the $\exists$-player has a winning strategy for the \qsatgame{} instance.
\end{proof}

Next, we show, how to replace edge costs not equal to $1$.

\begin{restatable}{theorem}{paththmuwwgPspaceHard}\label{path:thm:uwwgPspaceHard}
	\ostuwwg{} is \PSPACE-complete with unit edge costs.
\end{restatable}
% !TEX root=lipics-paper.tex
\begin{proof}
	In the reduction of \Cref{path:lem:uwWalkPspaceHard}, the only edge costs that are not $1$ are $M_1$ and $M_2$. 
	First we present two constructions to replace these edges.
	$M_1$ is larger than $k$, the length of the walk to be found.
	Thus, our construction needs to prevent the algorithm from using the edges that had cost $M_1$ to skip parts of variable and clause gadgets.\\
	All edges with cost $M_1$ reveal a new neighbor.
	Further, the next vertex the algorithm is supposed to choose when discovering a cost $M_1$ edge is a known one.
	Therefore, we can add $M_1$ many sinks to the vertex from where the cost $M_1$ edge would be discovered.
	This allows us to replace the cost $M_1$ edge with a cost $1$ edge.
	If such a vertex has multiple cost $M_1$ edges, the sinks only need to be added once.
	Then, the algorithm can only traverse the edge that had cost $M_1$ after discovering all sinks, as the adversary can force it to visit sinks first.
	Since $M_1 > k$, that results in a loss. 
	This also means that an unsatisfied clause now induces an additional cost of $M_1$ instead of $2$.
	The edges with cost $M_2$ can be traversed at most once without inducing too much cost to be able to win the game.
	Thus, we can simply replace them by a path of $M_2$ edges with cost $1$ each. \\
	It remains to consider the case, where the algorithm traverses an edge that had cost $M_1$ from the vertex that got identified.
	This is always possible, as the other vertex of the edge has already been visited and is thus identifiable.
	However, this only induces additional costs for the algorithm, as it visits vertices it has already visited, without being able to identify neighbors of those vertices that it has not visited already.
	In particular, if it traverses an edge 
	\begin{itemize}
		\item $\{u^i_3, p^i_1\}$, $\{u^i_4, p^i_j\}$, for $j \in \{2,3,4\}$, $\{p^i_1, u^i_5\}$, where $i \in \{1,\dots,|X|\}$, it has to traverse back to $u^i_5$ without being able to visit or discover new vertices.
		\item $\{p^i_2, v^i\}$ it is able to rechoose its assignment of variable $x_i$ before revealing parts of the clause gadgets, however since this edge only exists if $x_i$ is $\exists$-quantified, it was already able to choose the assignment without additional cost.
		\item $\{p^i_3, v^i_3\}$, $\{p^i_3, w^i_5\}$ or $\{p^i_3, u^i_7\}$, it cannot rechoose its assignment of variable $x_i$ (after revealing parts of the clause gadgets), because that requires traversing three paths of length $M_2$ in this variable gadget in total. 
		That forces a loss, since it is forced to traverse exactly two of these paths per variable gadget, and $k < 2\cdot |X| \cdot M_2$.
		Thus, it has to use one of those edges that connects to a known (and thus visited) vertex to leave the variable gadget.
		\item $\{v^i_2, c^j_i\}$ or $\{w^i_4, c^j_i\}$, it cannot reenter the clause gadgets through a different variable gadget, or assign the variable the opposite value as well, by the same arguments as above.
		It can only leave by returning to a clause gadget it has already visited, again not gaining additional information.
		\item $\{c^i_1, c^{i+1}_1\}$ for $i \in \{1,\dots,|C|-1\}$, it can only return to a clause gadget it has already visited, without additional information.
	\end{itemize}
	Thus by an analogous argument to \Cref{path:lem:uwWalkPspaceHard}, the claim follows.
\end{proof}

All our results for online games based on $s$-$t$-connectivity problems are summarized in \Cref{path:tab:results}.

\begin{table}[!ht]
	\centering
	\begin{tabular}{ll|l|l}
		\multicolumn{2}{l|}{}                                    & no costs                                                              & unit costs                       \\ \hline
		\multicolumn{1}{l|}{\multirow{2}{*}{Path}}  & undirected & \multirow{2}{*}{PSPACE-complete}                                      & \multirow{2}{*}{PSPACE-complete} \\ \cline{2-2}
		\multicolumn{1}{l|}{}                       & directed   &                                                                       &                                  \\ \hline
		\multicolumn{1}{l|}{\multirow{2}{*}{Trail}} & undirected & \multirow{2}{*}{PSPACE-complete}                                      & \multirow{2}{*}{PSPACE-complete} \\ \cline{2-2}
		\multicolumn{1}{l|}{}                       & directed   &                                                                       &                                  \\ \hline
		\multicolumn{1}{l|}{\multirow{2}{*}{Walk}}  & undirected & LOGSPACE-complete                                                            & \multirow{2}{*}{PSPACE-complete} \\ \cline{2-3}
		\multicolumn{1}{l|}{}                       & directed   &  
		PSPACE-complete&                                 
	\end{tabular}
	\caption{
		The complexity of different online games based on $s$-$t$-connectivity problems.
		The unit cost cases for path and trail are trivial consequences of our results assuming no edge costs.
	}
	\label{path:tab:results}
\end{table}

% !TEX root=lipics-paper.tex
% !TEX root=lipics-paper.tex
\section{Hamiltonian Problems}
\label{sec:hampath}
The second class of problems are Hamiltonian problems.
These problems can be interpreted as decision variants of Graph Exploration, which ask for a search sequence visiting all vertices of the graph.
The \otsg{} is the corresponding cost variant.
We start with the \ouhpg, which is defined as follows.

\begin{definition}[\ouhpg]\label{ham:def:ouhpg}\hfill\\
	\textbf{Given:} An undirected graph $G = (V, E)$, and a vertex $s \in V$. \\
	\textbf{Question:} Does an online algorithm exist, that finds a Hamiltonian path starting from $s$ in $G$ for all strategies of the adversary in the neighborhood reveal model, while the online algorithm knows an unlabeled map of $G$?
\end{definition}

We reduce \qsatgame{} to \ouhpg{} to show the \PSPACE-completeness.
Furthermore, we derive results for variations of this problem on $s$-$t$ path, $s$-$t$ trail and $s$-$t$ walk on directed and undirected graphs.

\begin{restatable}{theorem}{hamthmuHamPathPspaceHard}
	\label{ham:thm:uHamPathPspaceHard}
	\ouhpg{} is \PSPACE-complete.
\end{restatable}

\paragraph*{Reduction Overview}

The reduction is based on the book of Arora and Barak \cite{DBLP:books/daglib/0023084}.
They reduce \sat{} to \dhp{} by providing gadgets for the variables and the clauses.
We additionally use the reduction from \dhp{} to \uhp{} by Karp \cite{DBLP:conf/coco/Karp72} in order to provide the hardness for the undirected case.

The variable gadget is a chain of vertices, which has to be fully explored.
This chain can be traversed in two different variations.
These variations encode the decision on the variables either to be set to true or false.
Each clause gadget is essentially one vertex, which is connected to the variable gadget chains.
The connection is done in such a way that if and only if a non-negated variable is part of the clause, the traversal variation, which encodes the assignment to true, additionally allows to traverse the clause vertex.
Overall, the variable gadgets are concatenated together.
Moreover, a vertex $s$ is added and connected to the first vertex of the first variable gadget and a vertex $t$ is added and connected to the last vertex of the last variable gadget.
Thus, a Hamiltonian $s$-$t$-path assigns all variables a truth value and visits all clause vertices.

An important part of our reduction that uses the online nature of the game are additional edges to reveal vertices in the variable gadget.
Thus, it is possible to recognize them in order to traverse the variable gadget correctly and visit the clause vertices while finding back into the same variable gadget.

\paragraph*{Variable Gadget Overview}

The variable gadget for variable $x_i$ roughly consists of three parts, as in \Cref{ham:fig:variable}.

\begin{enumerate}
	\item A chain of vertices allows for traversing it in two ways.
			This part is the same as in the reduction from \sat{} to \uhp{} and it has the same functions.
			Beginning at $u^i_8$, either $v^i_1$ (assignment to true) or $w^i_1$ (assignment to false) can be visited.
			After that, the chain is traversed either along $u^i_8, v^i_1, m^i_1,  w^i_1, v^i_2, m^i_2,  w^i_2, \ldots,w^i_{4|C|}, u^i_9, u^i_0$ or along $u^i_8, w^i_1, m^i_1,  v^i_1, w^i_2, m^i_2,  v^i_2, \ldots,v^i_{4|C|}, u^i_9, \allowbreak u^i_0$.
			These sequences visit all vertices in the chain.
	\item This part identifies the middle vertices $m^i_j$ of the original variable gadget, such that the online algorithm is able to 						distinguish the $m^i_j$ from the clause vertices.
	\item The first part ensures that all vertices $v^i_j$ and $w^i_j$ and the vertex $u^i_7$ are revealed.
			Thus, it is possible to follow the chain, because the online algorithm can distiguish the middle vertices $m^i_j$ from the vertices $v^i_j$ and $w^i_j$.
			The edge to $u^i_7$ ensures that one can safely travel from $u^i_6$ to $u^i_7$ without traveling into the variable gadget and get stuck eventually.
			At last, if the variable is $\exists$-quantified, the edge to $v^i_1$ lets the online algorithm differentiate vertices $v^i_1$ and $w^i_1$.
			Thus, the online algorithm is able to choose the assignment by either travelling to $v^i_1$ (assignment to true) or $w^i_1$ (assignment to false).
\end{enumerate}

% !TEX root=lipics-paper.tex
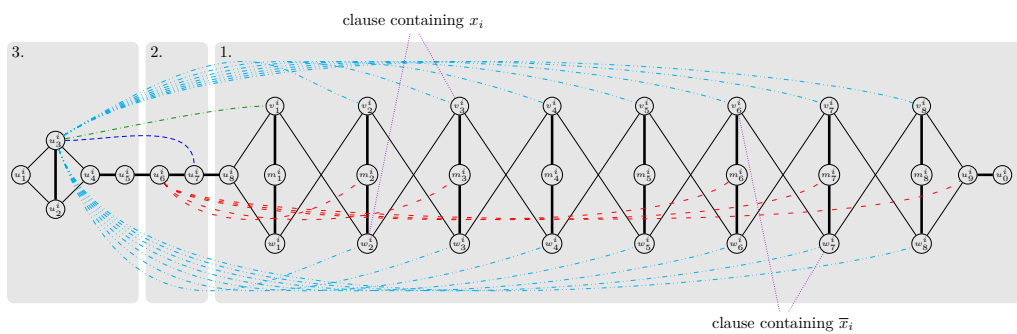
\begin{figure}[!ht]
	\centering
	\resizebox{\textwidth}{!}{
		\def\revealcolorid{blue}\def\revealcoloriddot{densely dashed}
		\def\revealcolortfe{dkgreen}\def\revealcolortfedot{dash dot}
		\def\revealcolorv{cyan}\def\revealcolorvdot{dash dot dot}
		\def\revealcolorm{red}\def\revealcolormdot{loosely dashed}
		\def\revealcolorw{cyan}\def\revealcolorwdot{dash dot dot}
		\def\clausecolor{mauve}\def\clausecolordot{densely dotted}
		\begin{tikzpicture}[scale=1,
			box/.style = {rectangle, fill=gray!20, rounded corners, fill opacity=1, inner sep=5pt}]
			\node (b11) at (4, 2.6) {};
			\node (b12) at (21, -2.5) {};
			\node[box, fit=(b11)(b12)] (box1) {};
			\node[below right, inner sep=3pt] at (box1.north west) {1.};
			\node (b31) at (2.5, 2.6) {};
			\node (b32) at (3.25, -2.5) {};
			\node[box, fit=(b31)(b32)] (box2) {};
			\node[below right, inner sep=3pt] at (box2.north west) {2.};
			\node (b41) at (-0.5, 2.6) {};
			\node (b42) at (1.75, -2.5) {};
			\node[box, fit=(b41)(b42)] (box3) {};
			\node[below right, inner sep=3pt] at (box3.north west) {3.};
			% --------------------------------------------------------------------------------
			% ID gadget
			% --------------------------------------------------------------------------------
			\node[shape=circle, draw=black, inner sep=0pt, minimum size=0.25cm] (id1) at (-0.5, 0) {\tiny$u^i_1$};
			\node[shape=circle, draw=black, inner sep=0pt, minimum size=0.25cm] (id2) at (0.25, 0.75) {\tiny$u^i_3$};
			\node[shape=circle, draw=black, inner sep=0pt, minimum size=0.25cm] (id3) at (0.25, -0.75) {\tiny$u^i_2$};
			\node[shape=circle, draw=black, inner sep=0pt, minimum size=0.25cm] (id4) at (1, 0) {\tiny$u^i_4$};
			\node[shape=circle, draw=black, inner sep=0pt, minimum size=0.25cm] (id5) at (1.75, 0) {\tiny$u^i_5$};
			\node[shape=circle, draw=black, inner sep=0pt, minimum size=0.25cm] (id6) at (2.5, 0) {\tiny$u^i_6$};
			\node[shape=circle, draw=black, inner sep=0pt, minimum size=0.25cm] (id7) at (3.25, 0) {\tiny$u^i_7$};

			\path[-] (id1) edge (id2) edge (id3) (id4) edge (id2) edge (id3);
			\path[-, ultra thick] (id2) edge (id3);
			\path[-, ultra thick] (id4) edge (id5) (id5) edge (id6) (id6) edge (id7);
			
			% --------------------------------------------------------------------------------
			% variable gadget
			% --------------------------------------------------------------------------------
			\pgfmathtruncatemacro{\a}{2}
			\node[shape=circle, draw=black, inner sep=0pt, minimum size=0.25cm] (v1d) at (4, 0) {\tiny$u^i_8$};
			\node[shape=circle, draw=black, inner sep=0pt, minimum size=0.25cm] (v1e) at (20, 0) {\tiny$u^i_9$};
			\node[shape=circle, draw=black, inner sep=0pt, minimum size=0.25cm] (v1e2) at (20.75, 0) {\tiny$u^i_{0}$};
			\path[-, ultra thick] (v1e) edge (v1e2);
			% paths of length 3
			\foreach \x in {1,...,8} {
				\foreach \y in {1,2,3} {
					\ifthenelse{\y = 1}{
						\def\name{\tiny$w^i_{\x}$}
					}{}
					\ifthenelse{\y = 2}{
						\def\name{\tiny$m^i_{\x}$}
					}{}
					\ifthenelse{\y = 3}{
						\def\name{\tiny$v^i_{\x}$}
					}{}
					
					\node[shape=circle, draw=black, inner sep=0pt, minimum size=0.25cm] (v1\x\y) at (3+\a*\x, 1.5*\y-3) {\name};
					\pgfmathtruncatemacro{\h}{\y-1}
					\ifthenelse{\h > 0}{
						\path[-, ultra thick] (v1\x\h) edge (v1\x\y);
					}{}
				}
				\pgfmathtruncatemacro{\h}{\x-1}
				\ifthenelse{\h > 0}{
					\ifthenelse{\h = 2 \OR \h = 6}{
						%\node[] (dot\h) at (\a*\x+2, 0) {};
						%\path[-, dashed] (v1\h1) edge (dot\h) edge (v1\x3) (v1\h3) edge (dot\h) edge (v1\x1);
						%\node[shape=rectangle, draw=gray!20, fill=gray!20, minimum size=1cm] (dot\h) at (\a*\x+2, 0) {$\ldots$};
						\path[-] (v1\h1) edge (v1\x3) (v1\h3) edge (v1\x1);
					}{
						\path[-] (v1\h1) edge (v1\x3) (v1\h3) edge (v1\x1);
					}
				}{}
			}
			% connect start and end of variable gadget
			\path[-, ultra thick] (id7) edge (v1d);
			\path[-] (v1d) edge (v111) edge (v113);
			\path[-] (v1e) edge (v181) edge (v183);
			
			% --------------------------------------------------------------------------------
			% reveal edges
			% --------------------------------------------------------------------------------
			% edge revealing t/f for exists
			\path[-, draw=\revealcolortfe, \revealcolortfedot, out=10, in=180, out distance=1cm, in distance=1cm] (id2) edge (v113);
			% edge reveal u7
			\path[-, draw=\revealcolorid, \revealcoloriddot, out=0, in=90, out distance=1cm, in distance=1cm] (id2) edge (id7);
			% edges revealing endpoints of 3-paths (for returning from clauses), except first
			\path[-, draw=\revealcolorv, \revealcolorvdot, out=32, in=160, out distance=5cm, in distance=2cm] (id2) edge (v123);
			\path[-, draw=\revealcolorv, \revealcolorvdot, out=32, in=160, out distance=5cm, in distance=2cm] (id2) edge (v133);
			\path[-, draw=\revealcolorv, \revealcolorvdot, out=32, in=160, out distance=5cm, in distance=2cm] (id2) edge (v143);
			\path[-, draw=\revealcolorv, \revealcolorvdot, out=32, in=160, out distance=5cm, in distance=2cm] (id2) edge (v153);
			\path[-, draw=\revealcolorv, \revealcolorvdot, out=32, in=160, out distance=5cm, in distance=2cm] (id2) edge (v163);
			\path[-, draw=\revealcolorv, \revealcolorvdot, out=32, in=160, out distance=5cm, in distance=2cm] (id2) edge (v173);
			\path[-, draw=\revealcolorv, \revealcolorvdot, out=32, in=160, out distance=5cm, in distance=2cm] (id2) edge (v183);
			\path[-, draw=\revealcolorw, \revealcolorwdot, out=290, in=200, out distance=5cm, in distance=2cm] (id2) edge (v121);
			\path[-, draw=\revealcolorw, \revealcolorwdot, out=290, in=200, out distance=5cm, in distance=2cm] (id2) edge (v131);
			\path[-, draw=\revealcolorw, \revealcolorwdot, out=290, in=200, out distance=5cm, in distance=2cm] (id2) edge (v141);
			\path[-, draw=\revealcolorw, \revealcolorwdot, out=290, in=200, out distance=5cm, in distance=2cm] (id2) edge (v151);
			\path[-, draw=\revealcolorw, \revealcolorwdot, out=290, in=200, out distance=5cm, in distance=2cm] (id2) edge (v161);
			\path[-, draw=\revealcolorw, \revealcolorwdot, out=290, in=200, out distance=5cm, in distance=2cm] (id2) edge (v171);
			\path[-, draw=\revealcolorw, \revealcolorwdot, out=290, in=200, out distance=5cm, in distance=2cm] (id2) edge (v181);
			% edges revealing middle vertices
			%\path[-, draw=\revealcolorm, out=300, in=220, out distance=0.8cm, in distance=2cm] (id6) edge (v112);
			\path[-, draw=\revealcolorm, \revealcolormdot, out=300, in=220, out distance=1cm, in distance=2cm] (id6) edge (v122);
			\path[-, draw=\revealcolorm, \revealcolormdot, out=300, in=220, out distance=1cm, in distance=2cm] (id6) edge (v132);
			%\path[-, draw=\revealcolorm, out=300, in=220, out distance=1cm, in distance=2cm] (id6) edge (v142);
			%\path[-, draw=\revealcolorm, out=300, in=220, out distance=1cm, in distance=2cm] (id6) edge (v152);
			\path[-, draw=\revealcolorm, \revealcolormdot, out=300, in=220, out distance=1cm, in distance=2cm] (id6) edge (v162);
			\path[-, draw=\revealcolorm, \revealcolormdot, out=300, in=220, out distance=1cm, in distance=2cm] (id6) edge (v172);
			%\path[-, draw=\revealcolorm, out=300, in=220, out distance=1cm, in distance=2cm] (id6) edge (v182);
			\path[-, draw=\revealcolorm, \revealcolormdot, out=300, in=220, out distance=1cm, in distance=2cm] (id6) edge (v1e);

			% --------------------------------------------------------------------------------
			% clause stubs
			% --------------------------------------------------------------------------------
			\foreach \x in {2} {
				\node (c\x) at (4+\a*\x, 3.35) {clause containing $x_i$};
				\node (c\x1) at (3.95+\a*\x, 3.2) {};
				\node (c\x2) at (4.05+\a*\x, 3.2) {};
				\pgfmathtruncatemacro{\h}{\x+1}
				\path[-, draw=\clausecolor, \clausecolordot] (c\x1) edge (v1\x1) (c\x2) edge (v1\h3);
			}
			\foreach \x in {6} {
				\pgfmathtruncatemacro{\c}{\x}
				\node (c\x) at (4+\a*\x, -3.25) {clause containing $\overline{x}_i$};
				\node (c\x1) at (3.95+\a*\x, -3.1) {};
				\node (c\x2) at (4.05+\a*\x, -3.1) {};
				\pgfmathtruncatemacro{\h}{\x+1}
				\path[-, draw=\clausecolor, \clausecolordot] (c\x1) edge (v1\x3) (c\x2) edge (v1\h1);
			}
		\end{tikzpicture}
	}
	\caption{Variable gadget for the reduction from \qsatgame{} to \ohpg.
			The thick solid black edges have to be traversed by the online algorithm.
			The thin solid black edges are optionally traversable.
			The dotted purple edges are used to simulate the satisfaction of a clause by traversing them to visit the clause vertex. 
			The dot-dashed green edge only exists if the variable is $\exists$-quantified. 
			All other non-solid colored edges cannot be used, and only exist for later recognition of the vertex, it is connected to.}
	\label{ham:fig:variable}
\end{figure}

\paragraph*{Formal Definition of the Variable Gadget}\label{ham:def:varGadget}
\begin{enumerate}
	\item The online algorithm enters at vertex $u^i_8$ and ends at $u^i_0$.
	In between are $4|C|$ paths of length three $v^i_j, m^i_j, w^i_j$, for $j \in \{1, \ldots, 4|C|\}$, whereby $v^i_k$ is connected to $w^i_{k+1}$ and $w^i_k$ is connected to $v^i_{k+1}$, for $k \in \{1, \ldots, 4|C|-1\}$.
	Besides, $u^i_8$ is connected to $v^i_1$ and $w^i_1$, $v^i_{4|C|}$ and $w^i_{4|C|}$ to $u^i_9$ and $u^i_9$ to $u^i_0$.
	\item This part consists of two vertices $u^i_6$ and $u^i_7$, whereby $u^i_6$ is connected to $u^i_7$ and $u^i_7$ is connected to $u^i_8$.
	Furthermore, $u^i_6$ is connected to $u^i_9$ and $m^i_j$ if and only if $j \mod 4 \geq 2$.
	\item There are five vertices $u^i_1, u^i_2, u^i_3, u^i_4, u^i_5$ with edges $\{u^i_1, u^i_2\}$, $\{u^i_1, u^i_3\}$, $\{u^i_2, u^i_3\}$, $\{u^i_2, u^i_4\}$, $\{u^i_3, u^i_4\}$, $\{u^i_4, u^i_5\}$.
	Additionally, $u^i_5$ is connected to $u^i_6$ and $u^i_3$ is connected to $u^i_7$ and to $v^i_j$ and $w^i_j$, for $j \in \{2, \ldots, 4|C|\}$.
	At last, if the variable is $\exists$-quantified, $u^i_3$ is connected to $v^i_1$.
\end{enumerate}

\paragraph*{Clause Gadget}\label{ham:def:clauseGadget}
The clause gadget for a clause $\{x_a, x_b, x_c\} = C_{k} \in C$ is one vertex $c_k$ connected with two edges to each of the variable gadgets of $x_a, x_b, x_c$.
Then, the two edges $\{c_{k}, w^i_{4k-2}\}$ and $\{c_{k}, v^i_{4k-1}\}$ are added, if the variable is non-negated, and $\{c_{k}, v^i_{4k-2}\}$ and $\{c_{k}, w^i_{4k-1}\}$, if the variable is negated.
The clause gadget is shown in \Cref{ham:fig:clause}.

% !TEX root=lipics-paper.tex
\begin{figure}[!ht]
	\centering
	\resizebox{0.7\textwidth}{!}{
		\def\sinkcolor{red}
		\def\revealcolorid{blue}
		\def\revealcolortfe{dkgreen}
		\def\revealcolortff{green}
		\def\revealcolorva{mauve}\def\revealcolorvadot{densely dotted}
		\def\revealcolorv{cyan}\def\revealcolorvdot{dash dot dot}
		\def\revealcolorvt{cyan}
		\def\revealcolorvf{orange}
		\def\clausecolor{gray}
		\begin{tikzpicture}[scale=0.6]
			% --------------------------------------------------------------------------------
			% clause gadget
			% --------------------------------------------------------------------------------
			\node[shape=circle, draw=black, inner sep=0pt, minimum size=0.25cm] (c11) at (0, 1) {};
			\node (c1) at (0, 2) {\huge $c_1$};
			
			% --------------------------------------------------------------------------------
			% variable stubs
			% --------------------------------------------------------------------------------
			\foreach \x in {1,2,3} {
				\pgfmathtruncatemacro{\hx}{\x-2}

                \ifthenelse{\x = 1}{
					\def\dx{a}
				}{}
				\ifthenelse{\x = 2}{
				    \def\dx{b}
				}{}
				\ifthenelse{\x = 3}{
					\def\dx{c}
				}{}
    
				% paths of length 3
				\foreach \y in {1,2} {
					\pgfmathtruncatemacro{\hy}{\y-1}
					\foreach \z in {1,2,3} {
						\ifthenelse{\y = 1}{
								\ifthenelse{\z = 1}{
									\def\name{\small$v^\dx_{4-2}$}
								}{}
								\ifthenelse{\z = 2}{
									\def\name{\small$m^\dx_{4-2}$}
								}{}
								\ifthenelse{\z = 3}{
									\def\name{\small$w^\dx_{4-2}$}
								}{}
						}{
								\ifthenelse{\z = 1}{
									\def\name{\small$v^\dx_{4-1}$}
								}{}
								\ifthenelse{\z = 2}{
									\def\name{\small$m^\dx_{4-1}$}
								}{}
								\ifthenelse{\z = 3}{
									\def\name{\small$w^\dx_{4-1}$}
								}{}
						}
						
						\pgfmathtruncatemacro{\hz}{\z-1}
						\node[shape=circle, draw=black, inner sep=0pt, minimum size=0.25cm, label=right:\name] (v\x\y\z) at (8*\hx+4*\hy-2, 10-2*\z) {};
						\ifthenelse{\z > 1}{
							\path[-, ultra thick] (v\x\y\hz) edge (v\x\y\z);
						}{}
					}
					% edges revealing how to return
					\foreach \a in {1,3} {
						\ifthenelse{\a > 2}{
							\ifthenelse{\y > 1}{\def\ha{1}}{\def\ha{-1}}
						}{\def\ha{0}}
						\node (vh\x\y\a) at (8*\hx+4*\hy-2+\ha, 11.5-2*\a) {};
						\path[-, very thick, draw=\revealcolorv, \revealcolorvdot] (vh\x\y\a) edge (v\x\y\a);
					}
				}
				% cross edges
				\path[-] (v\x13) edge (v\x21) (v\x11) edge (v\x23);
				% variable name
				\node (v\x) at (8*\hx, 9) {\large $x_\dx$};
			}
			
			% --------------------------------------------------------------------------------
			% connection to variables
			% --------------------------------------------------------------------------------
			\path[-, thick, draw=\revealcolorva, \revealcolorvadot, in=300, out=180, looseness=0.5] (c11) edge (v113);
			\path[-, thick, draw=\revealcolorva, \revealcolorvadot, out=160, in=240, out distance=7cm, in distance=5cm] (c11) edge (v121);
			\path[-, thick, draw=\revealcolorva, \revealcolorvadot, out=150, in=225, out distance=3cm, in distance=3cm] (c11) edge (v211);
			\path[-, thick, draw=\revealcolorva, \revealcolorvadot, out=30, in=260, out distance=1cm, in distance=1cm] (c11) edge (v223);
			\path[-, thick, draw=\revealcolorva, \revealcolorvadot, out=20, in=230] (c11) edge (v313);
			\path[-, thick, draw=\revealcolorva, \revealcolorvadot, out=10, in=255, out distance=9cm, in distance=7cm] (c11) edge (v321);
		\end{tikzpicture}
		}
	\caption{Clause gadget for the reduction from \qsatgame{} to \ohpg. Only parts of the variable gadgets are shown. The dash dot dotted cyan edges reveal the vertex that can be used to return to the variable gadget. The online algorithm is able to travel over the dotted purple edges to simulate the satisfaction of a clause.}
	\label{ham:fig:clause}
\end{figure}
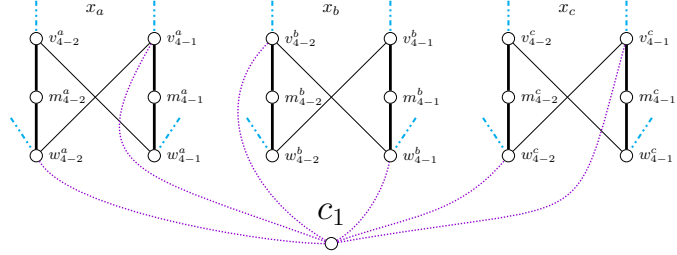

Now, we prove that the online algorithm traverses the variable gadget correctly while it is able to visit each clause vertex that is fulfilled by the assignment of the variable.
Additionally, we show that there is a one-to-one correspondence between the quantification of the variables in \qsat{} and the reduction.

First, we begin with the possible traversal order of each variable gadget.

\begin{restatable}{lemma}{hamlemvariable}\label{ham:lem:variable2}
	The online algorithm has to completely traverse Part 3 of the variable gadget, then Part 2, and then enter Part 1 through $u^i_8$ for the first time.
	In Part 1, it traverses the paths $v^i_j, m^i_j, w^i_j$ for $j \in \{1,\dots,4|C|\}$.
	If the variable is set to true, it traverses them in the order $v^i_j, m^i_j, w^i_j$, and otherwise in the order $w^i_j, m^i_j, v^i_j$, using the edge $\{w^i_j, v^i_{j+1}\}$ (resp. $\{v^i_j, w^i_{j+1}\}$) to get to the next path.
	If there is a vertex $c_{i_k}$ attached to $w^i_j$ and not $v^i_{j-1}$ (resp. $v^i_j$ and not $w^i_{j-1}$), it uses that vertex instead to reach the next path.
	Finally, it reaches $u^i_9$ and then leaves the variable gadget using $u^i_0$.
	Furthermore, it cannot leave a variable gadget via a clause vertex.
\end{restatable}
% !TEX root=lipics-paper.tex
\begin{proof}
	We consider the gadget part by part.
	\begin{description}
		\item[Part 3.] Assume the algorithm enters Part 2 of the variable gadget with at least one vertex of Part 3 not yet visited.
		There are three cases:
		It entered from $u^i_5$ and $u^i_3$ has been visited, it entered from $u^i_5$ and $u^i_3$ is not yet visited, or it entered from $u^i_3$ and $u^i_5$ is not visited. 
		
		In the first case, the vertex $u^i_2$ has not been visited and cannot be visited anymore.
		Thus the algorithm loses the game as it can no longer produce a Hamiltonian path.
		
		In the second case, it enters Part 2 in vertex $u^i_6$, and $u^i_7$ is still a new vertex.
		Then, the adversary can force the algorithm to visit $u^i_9$ next, as all neighbors of $u^i_6$ are new and thus indistinguishable.
		Since $u^i_3$ is not visited yet, all neighbors of $u^i_9$ are new and thus the adversary can force the algorithm to visit $u^i_0$ next.
		Then the remaining variable gadget has to be entered from one clause vertex $c_a$ and left through another clause vertex $c_b$.
		This partitions the variable gadget into three parts.
		The vertex $u^i_3$ can only be used once to jump between two parts, always forcing the algorithm to get stuck in the paths $v^i_1, m^i_1, w^i_1$ or $v^i_{4|C|}, m^i_{4|C|}, w^i_{4|C|}$, as they have no clause gadget attached, or to leave one of the parts unexplored.
		The same argument can then be applied recursively to the unexplored part.
		
		In the third case, it enters Part 2 in vertex $u^i_7$.
		Since both $u^i_6$ and $u^i_8$ are new in that case, the adversary can force the algorithm to visit $u^i_6$ next.
		Then $u^i_5$ is not visited yet, and visiting it makes the algorithm stuck in Part 3.
		On the other hand, not visiting $u^i_5$ produces a result that is not a Hamiltonian path.
		
		We do not need to consider the case of the algorithm skipping Part 2, as that can only be done from $u^i_3$, and then the adversary can always force the algorithm to enter $u^i_7$ and force a loss as described above.
		\item[Part 2.] Assume the algorithm leaves Part 2 from $u^i_6$. 
		Then, the adversary can force a loss by the same arguments as above.
		\item[Part 1.] If the algorithm deviates from the described path, the adversary can force a loss by the same argument as in the original reduction from \sat{} to \uhp{}, since all additional vertices of the variable gadget have already been visited and cannot be used again.
	\end{description}
	Finally, we show that the algorithm can actually find the described path:
	For Part 3 and Part 2 of the gadget, it can always choose a known neighbor when possible, which makes it end up in $u^i_8$.
	In Part 1, when the algorithm is on $v^i_j$ (resp. $w^i_j$) for $j \mod 4 \leq 1$, it has to choose the new vertex to reach $m^i_j$, unless it has already been visited, in which case it chooses the known neighbor $w^i_{j+1}$ (resp. $v^i_{j+1}$), or $u^i_9$ if $j = 4|C|$. 
	When the algorithm is on $v^i_j$ (resp. $w^i_j$) for $j \mod 4 \geq 2$ and there are two known neighbors, one of them is $m^i_j$ and has to be chosen, which can be done since it is the only possible choice that is a neighbor of $u^i_6$.
	Otherwise, if there is one known neighbor and a new neighbor, the new neighbor is a clause vertex and has to be chosen.
	Otherwise, if there are two known neighbors and one of them is a clause vertex, which has not been visited yet, it has to be chosen.
	From the clause vertex the algorithm can return to the known neighbor of the last step to continue traversing the variable gadget.
	When the algorithm is in a vertex $m^i_j$, there is always only one choice.
\end{proof}

Next, we prove that the quantification in the reduction graph corresponds to the quantification in the \qsat{}-instance.

\begin{restatable}{lemma}{hamlemvariableChoice}\label{ham:lem:variableChoice}
	If a variable $x_i$ is $\exists$-quantified, the online algorithm is able to choose its truth assignment.
	On the other hand if $x_i$ is $\forall$-quantified, the adversary is able to choose its truth assignment, and the online algorithm learns that truth assignment before choosing the next $\exists$-variable.
\end{restatable}
% !TEX root=lipics-paper.tex
\begin{proof}
	If the variable is $\exists$-quantified, the variable gadget contains the edge $\{u^i_3, v^i_1\}$.
	First, the online algorithm always visits $u^i_3$ and it always visits $u^i_8$.
	Due to the known edge $\{u^i_3, v^i_1\}$, the online algorithm is able to distinguish the known vertex $v^i_1$ from the new vertex $w^i_1$ when it is on $u^i_8$.
	Thus, the online algorithm has the power to decide the assignment.
	
	On the other hand, if the variable is $\forall$-quantified, the variable gadget does not contain the edge $\{u^i_3, v^i_1\}$.
	Then, both neighbors of $u^i_8$ are new.
	Thus, the adversary has the power to decide the assignment by sending the online algorithm to its preferred vertex $v^i_1$ or $w^i_1$.
	The online algorithm is able to identify the variable assignment before leaving the variable gadget from the map.
	Since the number of vertices it traverses before encountering the first clause vertex is different depending on the variable assignment, it can use the map to deduce which assignment it chose.
	Otherwise, the variable is not part of a clause and the assignment is irrelevant for solving the formula.
\end{proof}

At last, we show that all clause vertices are visited if and only if the online algorithm uses a truth assignment that fulfills all clauses.

\begin{restatable}{lemma}{hamlemclause}\label{ham:lem:clause2}
    The online algorithm is able to visit all vertices of $G$ if and only if for all $C_j \in C$ at least one variable $x_i \in C_j$ satisfies $C_j$ in the assignment chosen by the online algorithm.
\end{restatable}
% !TEX root=lipics-paper.tex
\begin{proof}
    Consider a clause $C_j \in C$.
	Let $x_i$ be a variable such that the assignment chosen by the algorithm satisfies the clause $C_j$.
	W.l.o.g., let the value assigned to $x_i$ be true.
	Then, the algorithm discovers $c_j$ as a new vertex or as a known neighbor from some other variable gadget when in $w^i_k$ for some $k \mod 4 = 2$ (the exact value of $k$ does not matter for the argument), where the only possible neighbors are $c_j$ and $v^i_{k+1}$ which is known from $u^i_3$.
	Thus, $c_j$ is identifiable as a clause vertex and can be visited, and according to \Cref{ham:lem:variable2}, the algorithm can also return to the correct position of the variable gadget.
	
	Now assume that $C_j$ is not satisfied.
	Then for any variable $x_i$ contained in $C_j$, when the algorithm discovers $c_j$, it was in a vertex $v^i_k$ (resp. $w^i_k$) which also had two known neighbors.
	Thus, visiting $c_j$ would have resulted in no longer being able to visit $m^i_k$.

    Together with \Cref{ham:lem:variable2}, we obtain that the whole graph can be visited if and only if all clauses are satisfied.
\end{proof}

\paragraph*{The Complete Reduction}

Given a \qsatgame{} instance with formula $\varphi$ containing variables $x_1, \dots, x_n$ and clauses $C_1, \dots, C_m$, we create an instance $G$ of \ouhpg{} as follows:
For each variable and clause, a gadget is created as described above and the clause gadgets are connected to the variable gadgets depending on the variables they contain.
For each $i \in \{1, \dots, n-1\}$, the edge $\{u^i_0, u^{i+1}_1\}$ is added, connecting all the variable gadgets.
Additionally, a vertex $s$ is added and connected to the first vertex of the first variable gadget and a vertex $t$ is added and connected to the last vertex of the last variable gadget.
An example of this construction can be seen in \Cref{ham:fig:complete}.
With this, we can now prove \Cref{ham:thm:uHamPathPspaceHard}.

% !TEX root=lipics-paper.tex
\begin{proof}[Proof of \Cref{ham:thm:uHamPathPspaceHard}]
	When connecting the variable gadgets as described above, the only possible next vertex after a vertex $u^i_0$ is $u^{i+1}_1$.
	Thus by \Cref{ham:lem:variable2}, the online algorithm can always find a path from $s$ to $t$.
	By \Cref{ham:lem:clause2}, that path is Hamiltonian if and only if all clauses are satisfied.
	Due to \Cref{ham:lem:variableChoice}, the online algorithm can decide the truth assignments of $\exists$-quantified variables, and the adversary can decide the truth assignments of $\forall$-quantified variables.
	Therefore the online algorithm can find a Hamiltonian path in $G$ if and only if the $\exists$-player has a winning strategy for the \qsatgame{} instance.
	The size of the variable gadgets is polynomially bounded by the number of clauses, while the clause gadgets have constant size.
	Thus, our reduction runs in polynomial time and the claim follows.
\end{proof}

% !TEX root=lipics-paper.tex
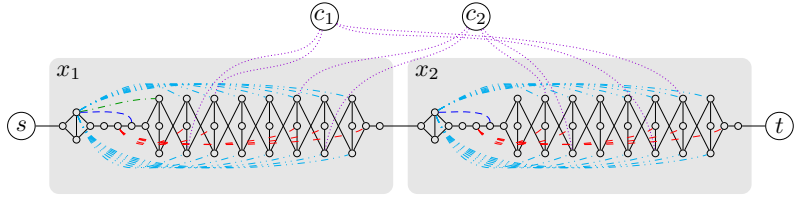
\begin{figure}[!ht]
	\centering
	\resizebox{0.8\textwidth}{!}{
		\def\revealcolorid{blue}\def\revealcoloriddot{densely dashed}
		\def\revealcolortfe{dkgreen}\def\revealcolortfedot{dash dot}
		\def\revealcolorv{cyan}\def\revealcolorvdot{dash dot dot}
		\def\revealcolorm{red}\def\revealcolormdot{loosely dashed}
		\def\revealcolorw{cyan}\def\revealcolorwdot{dash dot dot}
		\def\clausecolor{mauve}\def\clausecolordot{densely dotted}
		\begin{tikzpicture}[scale=0.4,
			node/.style = {shape=circle, draw, inner sep=0pt, minimum size=0.25cm},
			textnode/.style = {shape=circle, draw, inner sep=0pt, minimum size=0.4cm},
			smallnode/.style = {shape=circle, draw, inner sep=0pt, minimum size=0.1cm},
			box/.style = {rectangle, fill=gray!20, rounded corners, fill opacity=1, inner sep=2pt}]
			
			\def\a{1}
			\def\u{13}
			
			\foreach \v in {0,1} {
				\pgfmathtruncatemacro{\vv}{\v+1}
				\node (b\v0) at (\v*\u, 2) {};
				\node (b\v1) at (11.5+\v*\u, -2) {};
				\node[box, fit=(b\v0)(b\v1)] (box\v) {};
				\node[below right, inner sep=3pt] at (box\v.north west) {$x_{\vv}$};
				% --------------------------------------------------------------------------------
				% ID gadget
				% --------------------------------------------------------------------------------
				\node[smallnode] (\v1id1) at (\u*\v, 0) {};
				\node[smallnode] (\v1id2) at (0.5+\u*\v, 0.5) {};
				\node[smallnode] (\v1id3) at (0.5+\u*\v, -0.5) {};
				\node[smallnode] (\v1id4) at (1+\u*\v, 0) {};
				\node[smallnode] (\v1id5) at (1.5+\u*\v, 0) {};
				\node[smallnode] (\v1id6) at (2+\u*\v, 0) {};
				\node[smallnode] (\v1id7) at (2.5+\u*\v, 0) {};
	
				\path[-] (\v1id1) edge (\v1id2) edge (\v1id3) (\v1id4) edge (\v1id2) edge (\v1id3);
				\path[-] (\v1id2) edge (\v1id3);
				\path[-] (\v1id4) edge (\v1id5) (\v1id5) edge (\v1id6) (\v1id6) edge (\v1id7);
				
				% --------------------------------------------------------------------------------
				% variable gadget
				% --------------------------------------------------------------------------------
				\node[smallnode] (\v1v1d) at (3+\u*\v, 0) {};
				\node[smallnode] (\v1v1e) at (11+\u*\v, 0) {};
				\node[smallnode] (\v1v1e2) at (11.5+\u*\v, 0) {};
				\path[-] (\v1v1e) edge (\v1v1e2);
				% paths of length 3
				\foreach \x in {1,...,8} {
					\foreach \y in {1,2,3} {
						\node[smallnode] (\v1v1\x\y) at (2.5+\a*\x+\u*\v, \y-2) {};
						\pgfmathtruncatemacro{\h}{\y-1}
						\ifthenelse{\h > 0}{
							\path[-] (\v1v1\x\h) edge (\v1v1\x\y);
						}{}
					}
					\pgfmathtruncatemacro{\h}{\x-1}
					\ifthenelse{\h > 0}{
						\ifthenelse{\h = 2 \OR \h = 6}{
							%\node[] (dot\h) at (\a*\x+2, 0) {};
							%\path[-, dashed] (v1\h1) edge (dot\h) edge (v1\x3) (v1\h3) edge (dot\h) edge (v1\x1);
							%\node[shape=rectangle, draw=gray!20, fill=gray!20, minimum size=1cm] (dot\h) at (\a*\x+2, 0) {$\ldots$};
							\path[-] (\v1v1\h1) edge (\v1v1\x3) (\v1v1\h3) edge (\v1v1\x1);
						}{
							\path[-] (\v1v1\h1) edge (\v1v1\x3) (\v1v1\h3) edge (\v1v1\x1);
						}
					}{}
				}
				% connect start and end of variable gadget
				\path[-] (\v1id7) edge (\v1v1d);
				\path[-] (\v1v1d) edge (\v1v111) edge (\v1v113);
				\path[-] (\v1v1e) edge (\v1v181) edge (\v1v183);
				
				% --------------------------------------------------------------------------------
				% reveal edges
				% --------------------------------------------------------------------------------
				% edge revealing t/f for exists
				\ifthenelse{\v = 0}{
					\path[-, draw=\revealcolortfe, \revealcolortfedot, out=10, in=180, out distance=1cm, in distance=1cm] (\v1id2) edge (\v1v113);
				}{}
				% edge reveal u7
				\path[-, draw=\revealcolorid, \revealcoloriddot, out=0, in=90, out distance=0.5cm, in distance=0.5cm] (\v1id2) edge (\v1id7);
				% edges revealing endpoints of 3-paths (for returning from clauses), except first
				\path[-, draw=\revealcolorv, \revealcolorvdot, out=32, in=160, out distance=3cm, in distance=1cm] (\v1id2) edge (\v1v123);
				\path[-, draw=\revealcolorv, \revealcolorvdot, out=32, in=160, out distance=3cm, in distance=1cm] (\v1id2) edge (\v1v133);
				\path[-, draw=\revealcolorv, \revealcolorvdot, out=32, in=160, out distance=3cm, in distance=1cm] (\v1id2) edge (\v1v143);
				\path[-, draw=\revealcolorv, \revealcolorvdot, out=32, in=160, out distance=3cm, in distance=1cm] (\v1id2) edge (\v1v153);
				\path[-, draw=\revealcolorv, \revealcolorvdot, out=32, in=160, out distance=3cm, in distance=1cm] (\v1id2) edge (\v1v163);
				\path[-, draw=\revealcolorv, \revealcolorvdot, out=32, in=160, out distance=3cm, in distance=1cm] (\v1id2) edge (\v1v173);
				\path[-, draw=\revealcolorv, \revealcolorvdot, out=32, in=160, out distance=3cm, in distance=1cm] (\v1id2) edge (\v1v183);
				\path[-, draw=\revealcolorw, \revealcolorwdot, out=290, in=200, out distance=3cm, in distance=1cm] (\v1id2) edge (\v1v121);
				\path[-, draw=\revealcolorw, \revealcolorwdot, out=290, in=200, out distance=3cm, in distance=1cm] (\v1id2) edge (\v1v131);
				\path[-, draw=\revealcolorw, \revealcolorwdot, out=290, in=200, out distance=3cm, in distance=1cm] (\v1id2) edge (\v1v141);
				\path[-, draw=\revealcolorw, \revealcolorwdot, out=290, in=200, out distance=3cm, in distance=1cm] (\v1id2) edge (\v1v151);
				\path[-, draw=\revealcolorw, \revealcolorwdot, out=290, in=200, out distance=3cm, in distance=1cm] (\v1id2) edge (\v1v161);
				\path[-, draw=\revealcolorw, \revealcolorwdot, out=290, in=200, out distance=3cm, in distance=1cm] (\v1id2) edge (\v1v171);
				\path[-, draw=\revealcolorw, \revealcolorwdot, out=290, in=200, out distance=3cm, in distance=1cm] (\v1id2) edge (\v1v181);
				% edges revealing middle vertices
				%\path[-, draw=\revealcolorm, out=300, in=220, out distance=0.8cm, in distance=2cm] (id6) edge (v112);
				\path[-, draw=\revealcolorm, \revealcolormdot, out=300, in=220, out distance=1cm, in distance=1cm] (\v1id6) edge (\v1v122);
				\path[-, draw=\revealcolorm, \revealcolormdot, out=300, in=220, out distance=1cm, in distance=1cm] (\v1id6) edge (\v1v132);
				%\path[-, draw=\revealcolorm, out=300, in=220, out distance=1cm, in distance=2cm] (id6) edge (v142);
				%\path[-, draw=\revealcolorm, out=300, in=220, out distance=1cm, in distance=2cm] (id6) edge (v152);
				\path[-, draw=\revealcolorm, \revealcolormdot, out=300, in=220, out distance=1cm, in distance=1cm] (\v1id6) edge (\v1v162);
				\path[-, draw=\revealcolorm, \revealcolormdot, out=300, in=220, out distance=1cm, in distance=1cm] (\v1id6) edge (\v1v172);
				%\path[-, draw=\revealcolorm, out=300, in=220, out distance=1cm, in distance=2cm] (id6) edge (v182);
				\path[-, draw=\revealcolorm, \revealcolormdot, out=300, in=220, out distance=1cm, in distance=1cm] (\v1id6) edge (\v1v1e);
			}
			
			% --------------------------------------------------------------------------------
			% clauses
			% --------------------------------------------------------------------------------
			\def\x{2}
			\def\h{3}
			\node[textnode] (c\x) at (9.5, 4) {$c_1$};
			\path[-, draw=\clausecolor, \clausecolordot, in=80, out=250, in distance=5cm, out distance=2cm] (c\x) edge (01v1\x1);
			\path[-, draw=\clausecolor, \clausecolordot, in=80, out=260, in distance=1.5cm, out distance=2cm] (c\x) edge (01v1\h3);
			\path[-, draw=\clausecolor, \clausecolordot, in=120, out=330, in distance=3cm, out distance=1cm] (c\x) edge (11v173);
			\path[-, draw=\clausecolor, \clausecolordot, in=100, out=326, in distance=5.7cm, out distance=2cm] (c\x) edge (11v161);
			\def\x{6}
			\def\h{7}
			\node[textnode] (c\x) at (15, 4) {$c_2$};
			\path[-, draw=\clausecolor, \clausecolordot, in=70, out=230, in distance=2cm, out distance=2cm] (c\x) edge (01v1\x3);
			\path[-, draw=\clausecolor, \clausecolordot, in=80, out=240, in distance=4cm, out distance=2cm] (c\x) edge (01v1\h1);
			\path[-, draw=\clausecolor, \clausecolordot, in=100, out=300, in distance=4cm, out distance=1cm] (c\x) edge (11v131);
			\path[-, draw=\clausecolor, \clausecolordot, in=90, out=290, in distance=1cm, out distance=1cm] (c\x) edge (11v123);
		
			% --------------------------------------------------------------------------------
			% s t
			% --------------------------------------------------------------------------------
			\node[textnode] (s) at (-1.5,0) {$s$};
			\node[textnode] (t) at (26,0) {$t$};
			\path[-] (01v1e2) edge (11id1);
			\path[-] (s) edge (01id1);
			\path[-] (11v1e2) edge (t);
		\end{tikzpicture}
	}
	\caption{Sketch of the full construction of the reduction from \qsatgame{} to \ouhpg{} with the formula $\exists x_1 \forall x_2 (x_1 \vee x_2) \wedge (\overline{x}_1 \vee \overline{x}_2)$.}
	\label{ham:fig:complete}
\end{figure}

\paragraph*{Directed Graphs}
	
Next we show that the above ideas can also be applied to directed graphs. 
The \odhpg{} is defined analogously to the \ouhpg.
	
\begin{restatable}{corollary}{hamcordHamPathPspaceHard}\label{ham:cor:dHamPathPspaceHard}
	\odhpg{} is \PSPACE-complete.
\end{restatable}
% !TEX root=lipics-paper.tex
\begin{proof}
	We use a similar construction as for \Cref{ham:thm:uHamPathPspaceHard}, but replace the undirected edges with directed arcs as follows:
	\begin{itemize}
		\item Any edge that only reveals vertices but cannot be used by the arguments of \Cref{ham:lem:variable2} is replaced with an arc directed towards the revealed vertex.
		\item The edges $\{u^i_2, u^i_3\}$ and $\{v^i_j, m^i_j\}, \{m^i_j, w^i_j\}$, for $j \in \{1, \dots, 4|C|\}$ and $i \in \{1, \dots, |X|\}$, can be traversed in both directions, and are therefore replaced by arcs $(u^i_2, u^i_3)$, $(u^i_3, u^i_2)$, $(v^i_j, m^i_j)$, $(m^i_j, v^i_j)$, $(m^i_j, w^i_j)$ and $(w^i_j, m^i_j)$.
		\item The edges $\{v^i_j, c_k\}$ and $\{w^i_j, c_k\}$, for $j\ \text{mod}\ 4 = 2$, (if they exist) are replaced by arcs $(v^i_j, c_k)$ and $(w^i_j, c_k)$.
		\item The edges $\{v^i_j, c_k\}$ and $\{w^i_j, c_k\}$, for $j\ \text{mod}\ 4 = 3$, (if they exist) are replaced by arcs $(c_k, v^i_j)$ and $(c_k, w^i_j)$.
		\item The remaining edges can only be traversed in one direction (by the arguments of \Cref{ham:lem:variable2,ham:thm:uHamPathPspaceHard}) and are therefore replaced by the respective directed arc.
	\end{itemize}
	Then the claim follows from the same arguments as in \Cref{ham:lem:variable2,ham:lem:variableChoice,ham:lem:clause2,ham:thm:uHamPathPspaceHard}.
\end{proof}

\subsection*{Hamiltonian Cycle}\label{hamcycle:def:startGadget}

In this section, we consider online games based on \hc{} instead of \hp.
In the directed case, \PSPACE-completeness easily follows from \Cref{ham:cor:dHamPathPspaceHard} by just adding an arc from the last vertex of the last variable gadget back to the vertex the algorithm starts in.
However, when considering undirected graphs, we have to make sure that the algorithm traverses our construction in the correct direction.
For that, we introduce a start gadget, that allows the algorithm to choose a direction, and to simplify our arguments, also allows the adversary to force a loss if the algorithm chooses the wrong direction.

\paragraph*{Start Gadget for Hamiltonian Cycle}
The start gadget consists of two parts, the first allows the algorithm to choose a direction, while the second part allows the adversary to force a loss if the algorithm chose to traverse the construction backwards.
\begin{enumerate}
	\item There are vertices $s, s', t$. 
	Further, there are two sets $V_1 = \{v_{1,1}, v_{1,2}\}$ and $V_2 = \{v_{2,1}, v_{2,2}\}$. 
	$s$, $V_2$ and $v_{1,1}$ as well as $s$, $V_2$ and $v_{1,2}$ form two cliques of size $4$ that overlap in $s$ and $V_2$.
	$s'$ is connected to $V_2$, and $t$ is connected to $V_1$.
	\item There are vertices $v_3, v_4, v_5, v_6$ and $s''$.
	Vertices $v_3$ and $v_4$ are both connected to $s'$ and $v_5$, as well as to each other.
	Finally, there are the edges $\{v_5, v_6\}, \{v_6, s''\}, \{v_3, s''\}$ and $\{v_6, t\}$.
\end{enumerate}
An example of this construction is shown in \Cref{ham:fig:startDirection}.

% !TEX root=lipics-paper.tex
\begin{figure}[!ht]
	\centering
	\resizebox{0.7\textwidth}{!}{
		\begin{tikzpicture}[scale=1,
			box1/.style = {rectangle, fill=gray!20, rounded corners, fill opacity=1, inner sep=23pt},
			box2/.style = {rectangle, fill=gray!40, rounded corners, fill opacity=1, inner sep=17pt}]
			
			\node[] (P11) at (-4, -1) {};
			\node[] (P12) at (4, 1) {};
			\node[] (P21) at (6, -1) {};
			\node[] (P22) at (10, 1) {};
			
			\node[box1, fit=(P11)(P12)] (box1) {};
			\node[box1, fit=(P21)(P22)] (box2) {};
			\node[below right, inner sep=3pt] (P1) at (box1.north west) {$1.$};
			\node[below right, inner sep=3pt] (P2) at (box2.north west) {$2.$};
			
			\node[] (V311) at (-2, 1) {};
			\node[] (V322) at (-2, -1) {};
			\node[] (V411) at (2, 1) {};
			\node[] (V422) at (2, -1) {};

			\node[box2, fit=(V311)(V322)] (box3) {};
			\node[box2, fit=(V411)(V422)] (box4) {};
			\node[below right, inner sep=3pt] (V3) at (box3.north west) {$V_1$};
			\node[below right, inner sep=3pt] (V4) at (box4.north west) {$V_2$};

			\node[shape=circle, draw=black, inner sep=0pt, minimum size=0.75cm] (s) at (0, 0) {$s$};
			\node[shape=circle, draw=black, inner sep=0pt, minimum size=0.75cm] (s') at (4, 0) {$s'$};
			\node[shape=circle, draw=black, inner sep=0pt, minimum size=0.75cm] (t) at (-4, 0) {$t$};
			\node[shape=circle, draw=black, inner sep=0pt, minimum size=0.75cm] (v31) at (-2, 1) {$v_{1,1}$};
			\node[shape=circle, draw=black, inner sep=0pt, minimum size=0.75cm] (v32) at (-2, -1) {$v_{1,2}$};
			\node[shape=circle, draw=black, inner sep=0pt, minimum size=0.75cm] (v41) at (2, 1) {$v_{2,1}$};
			\node[shape=circle, draw=black, inner sep=0pt, minimum size=0.75cm] (v42) at (2, -1) {$v_{2,2}$};
			\node[shape=circle, draw=black, inner sep=0pt, minimum size=0.75cm] (v3) at (6, 1) {$v_3$};
			\node[shape=circle, draw=black, inner sep=0pt, minimum size=0.75cm] (v4) at (6, -1) {$v_4$};
			\node[shape=circle, draw=black, inner sep=0pt, minimum size=0.75cm] (v5) at (7, 0) {$v_5$};
			\node[shape=circle, draw=black, inner sep=0pt, minimum size=0.75cm] (v6) at (8.5, 0) {$v_6$};
			\node[shape=circle, draw=black, inner sep=0pt, minimum size=0.75cm] (s'') at (10, 0) {$s''$};

			\path[-] (s) edge (v31) edge (v32) edge (v41) edge (v42);
			
			\path[-] (v31) edge (v41) (v32) edge (v42) (v41) edge (v42);
			\path[-, in=170, out=280] (v31) edge (v42);
			\path[-, in=260, out=10] (v32) edge (v41);

			\path[-] (s') edge (v41) edge (v42) edge (v3) edge (v4);
			\path[-] (v5) edge (v3) edge (v4) edge (v6);
			\path[-] (v3) edge (v4) (v6) edge (s'');
			\path[-, in=135, out=0] (v3) edge (s'');
			
			\path[-, out=300, in=249, out distance=3cm, in distance=3cm] (t) edge (v6);

			\path[-] (t) edge (v31) edge (v32);
		\end{tikzpicture}
	}
	\caption{
		Start gadget for the reduction from \qsatgame{} to \ohcg.
		A possible correct traversal is $s, v_{2,1}, s', v_4, v_3, v_5, v_6, s''$, and after traversing the reduction graph for Hamiltonian path $t, v_{1,1}, v_{2,2}, v_{1,2}, s$.
	}
	\label{ham:fig:startDirection}
\end{figure}
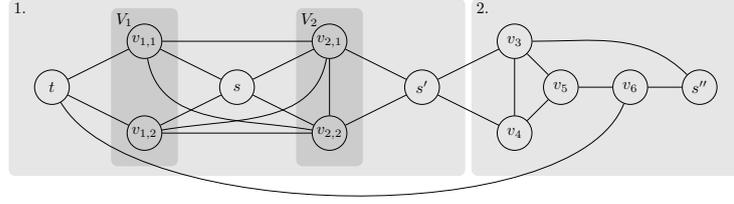

\begin{restatable}{lemma}{hamcyclelemstart}\label{hamcycle:lem:start}
	Let $G = (V, E)$ be a connected graph with two vertices $a, b \in V$, and let $a$ be connected to $s''$ and $b$ connected to $t$.
	Then, the algorithm can find a path from $s$ to $s''$, then travel from $a$ to $b$, and at last it can visit all remaining vertices of the start gadget from $t$ and return to $s$.
	Further, if the algorithm enters $t$ before $s''$, the adversary can force a loss.
\end{restatable}
% !TEX root=lipics-paper.tex
\begin{proof}
	First, we show that the algorithm can find a path from $s$ to $s''$, and then that it can visit all remaining vertices from $t$ and return to $s$.
	The algorithm has to go to a vertex $v \in N(s)$, whereby the adversary can choose between vertices from $V_1$ and $V_2$.
	The algorithm can distinguish the vertices of $V_1$ and $V_2$ by their degree.
	\begin{description}
		\item[Case 1] $v \in V_1$\\
		Then, the algorithm is able to choose a vertex in $V_2$ as those are distinguishable from $t$.
		Next, the algorithm can choose to visit $s'$, as it is a new vertex.
		\item[Case 2] $v \in V_2$\\
		Then, the algorithm can choose to visit $s'$, as it is a new vertex.
	\end{description}
	From $s'$ the algorithm can visit $v_3$ and $v_4$ in any order (from $s'$ they are both new, and after choosing one the other is the only known neighbor).
	If $v_3$ is visited second, $v_5$ is distinguishable from $s''$ as $v_5$ is known.
	If the algorithm chooses $s''$ instead of $v_5$, the algorithm loses the game as it can no longer visit $v_5$ and still return to $s$.
	Otherwise $v_5$ is the only choice and from there $s''$ can be reached as it is distinguishable from $t$ due to the known neighbor $v_3$.
	
	When returning to $t$ we distinguish the cases how the algorithm reached $s'$:
	\begin{description}
		\item[Case 1] From $t$, the algorithm can only visit the not yet visited vertex of $V_1$. 
		Then, it can choose to visit the not yet visited vertex of $V_2$, as it is the only not visited vertex left, and finally return to $s$.
		\item[Case 2] From $t$, the algorithm can visit either vertex of $V_1$. 
		Then, it can choose to visit the not yet visited vertex of $V_2$, as it is the only not visited neighbor, then visit the other vertex of $V_2$ and finally return to $s$.
	\end{description}
	Finally, the adversary can force a loss if the algorithm visits $t$ before $s''$:
	If the algorithm visits $t$ from $v_6$, it cannot visit the vertices of $G$ and return to $s$ anymore, since it already visited vertices on the only two paths from $s$ to vertices of $G$.
	On the other hand, if the algorithm chooses to enter $t$ in Case 1 above, the adversary can force it to visit $v_6$ next, as both neighbors of $t$ are new.
	From there, it can force $s''$ and $v_3$ next for the same reason, resulting in the same outcome as above.
\end{proof}

With the help of the starting gadget, we can reuse the construction from \Cref{ham:thm:uHamPathPspaceHard}.

\begin{restatable}{theorem}{hamcyclethmucgPspaceHard}\label{hamcycle:thm:ucgPspaceHard}
	\ouhcg{} is \PSPACE-complete.
\end{restatable}
% !TEX root=lipics-paper.tex
\begin{proof}
	We use the construction of \Cref{ham:thm:uHamPathPspaceHard}, but instead of attaching vertices $s$ and $t$ to the vertices $u^i_1$ and $u^{|X|}_0$, we add the start gadget and connect $s''$ with $u^i_1$ and $t$ with $u^{|X|}_0$.
	Then, the new start vertex for the game is $s$.
	Thus, the claim follows from the arguments of \Cref{hamcycle:lem:start,ham:thm:uHamPathPspaceHard}.
\end{proof}

\subsection{Relaxing the Path Constraint}

If the online algorithm is allowed to visit vertices or edges multiple times instead of precisely once, we say that the path constraint is relaxed to a walk or trail. 
We show that the solution for the relaxed version stays the same, i.e., it is always a cycle that visits every vertex precisely once, when we add an additional threshold of $|V|$, for the cost of the solution. 

\begin{restatable}{lemma}{hamlemwalkallvertices}\label{ham:lem:walkAllVertices}
	Let $G = (V, E, cost)$ be any graph with unit cost edges.
	Then any solution to \otsg{} on $G$ with relaxed path constraint to walk (resp. trail) is a cycle, if the cost threshold is $|V|$.
	This result can be extended to complete graphs by adding edges with cost $(1+\alpha)$ for $\alpha > 0$ for any non-existing edge in $G$.
\end{restatable}
% !TEX root=lipics-paper.tex
\begin{proof}
	Assume, the cyclic walk $w$ starting at $s \in V$ is not a cyclic path starting at $s$.
	Then, at least one vertex $v \neq s$ was visited twice in $w$.
	Thus, the walk includes the subsequences $u, v, x$ and $u', v, x'$.
	That is, $v$ is included twice.
	Because the walk is only able to travel exactly $|V|$ unit cost edges, the walk $w$ includes $|V|+1$ vertices at maximum, whereby $s$ has to be included twice.
	Consequently, there is a vertex $v'$,which is not included.
	Thus, $w$ is not visiting all vertices in $V$ and $w$ is no solution to $G$.
	
	The case if a trail is to be sought is analogous.
	Then, the starting vertex $s$ is included only once and at most $|V|$ edges can be traversed.
\end{proof}

In \Cref{ham:tab:results}, our results for online games based on Hamiltonian $s$-$t$-path problems are presented.
The fields with ? remain unanswered.

\begin{table}[!ht]
	\centering
	\begin{tabular}{ll|l|l}
		\multicolumn{2}{l|}{}                                    & no costs                                                              & unit costs                       \\ \hline
		\multicolumn{1}{l|}{\multirow{2}{*}{Path}}  & undirected & \multirow{2}{*}{PSPACE-complete}                                      & \multirow{2}{*}{PSPACE-complete} \\ \cline{2-2}
		\multicolumn{1}{l|}{}                       & directed   &                                                                       &                                  \\ \hline
		\multicolumn{1}{l|}{\multirow{2}{*}{Trail}} & undirected & \multirow{2}{*}{?}                                      & \multirow{2}{*}{PSPACE-complete} \\ \cline{2-2}
		\multicolumn{1}{l|}{}                       & directed   &                                                                       &                                  \\ \hline
		\multicolumn{1}{l|}{\multirow{2}{*}{Walk}}  & undirected & LOGSPACE-complete                                                            & \multirow{2}{*}{PSPACE-complete} \\ \cline{2-3}
		\multicolumn{1}{l|}{}                       & directed   &  
		?&                                 
	\end{tabular}
	\caption{
		The complexity of different \ohpg\textsc{s}.
	}
	\label{ham:tab:results}
\end{table}

% !TEX root=lipics-paper.tex
\section{Additional Problems}
\label{sec:additional}
There are many problems that are closely related to the $s$-$t$ path problem or the Hamiltonian path problem. 
Often, they only add simple constraints on the solution.
Our previously presented reductions can handle most of these variations by extending them with a small construction or setting a value, representing the constraint, to a specific number.
For all of the following problems, we use the definitions as given in the appendix of \cite{DBLP:books/fm/GareyJ79}.

\begin{restatable}{theorem}{addthmpathVariants}\label{add:thm:pathVariants}
	\opwfpg, \ocspg, \otdpg, \ovdpg{} are \PSPACE-complete.
\end{restatable}
% !TEX root=lipics-paper.tex
\begin{proof}
	The corresponding reductions are all directly derived from the corresponding reductions for \ostpg{}.
	The additional constraints on the path can be relaxed such that only a constant number of paths are to be found.
	One of those paths is the reduction path of \ostpg{}.
	The \PSPACE-completeness of all problems holds also for the trail and walk version of these problems as well as for directed and undirected graphs as proven in \Cref{sec:path}.
\end{proof}
Based on this proof, it is possible to argue that online games that require to compute an $s$-$t$ path in a graph, optionally with additional constraints, are \PSPACE-hard as one can drop the additional constraints on the $s$-$t$ path.
This is also true for $s$-$t$ trail and for $s$-$t$ walk problems as long as the base problem is \PSPACE-hard.

\begin{restatable}{corollary}{addcorstackerCrane}\label{add:cor:stackerCrane}
	\oscg{} is \PSPACE-complete.
\end{restatable}
% !TEX root=lipics-paper.tex
\begin{proof}
	We use the construction for \ostupg{} with the following small modification:
	The arc $(t, s)$ is added.
	Since the starting point is still $s$ and the algorithm needs to traverse all arcs, it has to find an $s$-$t$ path to be able to traverse the added arc.
	Thus, the claim follows from the arguments of \Cref{path:thm:uPathPspaceHard}.
\end{proof}

\begin{restatable}{theorem}{addthmruralPostman}\label{add:thm:ruralPostman}
	\orpg{} is \PSPACE-complete.
\end{restatable}

% !TEX root=lipics-paper.tex
\begin{proof}
	We combine the start gadget for the \ouhcg{} with the reduction for \ostupg:
	The vertex $s$ of the construction from \Cref{path:thm:uPathPspaceHard} is replaced with the start gadget for \ouhcg, and the vertex $s''$ is connected to $u^1_1$ instead.
	An additional vertex $v$ is added, and is connected to the vertex $t$ of the start gadget as well as $u^{|X|}_7$.
	The only edge that the algorithm is forced to traverse is $\{v, u^{|X|}_7\}$, and the starting vertex is $s$ (of the start gadget).
	An example of this construction can be seen in \Cref{ham:fig:startRural}.
	
	By the arguments of \Cref{hamcycle:lem:start}, the algorithm has to leave the start gadget via $s''$, as otherwise the adversary can force a loss.
	Thus, it can reach $u^{|X|}_7$ and the edge $\{v, u^{|X|}_7\}$ if and only if it finds a path from $u^1_1$ to $u^{|X|}_7$.
	Therefore, the claim follows from \Cref{path:thm:uPathPspaceHard}.
\end{proof}
% !TEX root=lipics-paper.tex
\begin{figure}[!ht]
	\centering
	\resizebox{\textwidth}{!}{
		\begin{tikzpicture}[scale=1,
			box/.style = {rectangle, fill=gray!20, rounded corners, fill opacity=1, inner sep=5pt},
			node/.style = {shape=circle, draw=black, inner sep=0pt, minimum size=0.9cm}]

			\node[node, minimum size=0.9cm] (s) at (0, 0) {\LARGE$s$};
			\node[node] (s') at (4, 0) {};
			\node[node, minimum size=0.9cm] (t) at (-4, 0) {\LARGE$t$};
			\node[node] (v31) at (-2, 1) {};
			\node[node] (v32) at (-2, -1) {};
			\node[node] (v41) at (2, 1) {};
			\node[node] (v42) at (2, -1) {};
			\node[node] (v3) at (6, 1) {};
			\node[node] (v4) at (6, -1) {};
			\node[node] (v5) at (7, 0) {};
			\node[node] (v6) at (8.5, 0) {};
			\node[node] (s'') at (10, 0) {\LARGE$s''$};
			
			\path[-] (s) edge (v31) edge (v32) edge (v41) edge (v42);
			
			\path[-] (v31) edge (v41) (v32) edge (v42) (v41) edge (v42);
			\path[-, in=170, out=280] (v31) edge (v42);
			\path[-, in=260, out=10] (v32) edge (v41);
			
			\path[-] (s') edge (v41) edge (v42) edge (v3) edge (v4);
			\path[-] (v5) edge (v3) edge (v4) edge (v6);
			\path[-] (v3) edge (v4) (v6) edge (s'');
			\path[-, in=135, out=0] (v3) edge (s'');
			
			\path[-, out=300, in=249, out distance=3cm, in distance=3cm] (t) edge (v6);
			
			\path[-] (t) edge (v31) edge (v32);

			\node[node] (t1) at (22, 0) {\Large$u_7^{|X|}$};
			
			\node[node, minimum size=0.9cm] (21) at (-6, 0) {\LARGE$v$};
			\path[-] (21) edge (t);

			\path[-, in=300, out=240, looseness=0.4, dotted, draw=blue] (t1) edge (21);
			
			\node[align=center] (text) at (16, 0) {{\LARGE Reduction for}\\ {\LARGE \ostupg{}.}};

		\end{tikzpicture}
	}
	\caption{Reduction from \qsatgame{} to \orpg. The dotted blue edge has to be traversed by the algorithm.}
	\label{ham:fig:startRural}
\end{figure}
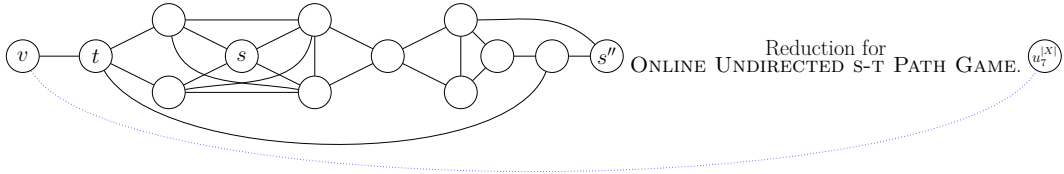
	
Additionally, we derive hardness results for problems based on \hp.
\begin{restatable}{theorem}{addthmhamVariants}\label{add:thm:hamVariants}
	\omtsg, \obtsg, \olcg, \olpg{} are \PSPACE-complete.
\end{restatable}
% !TEX root=lipics-paper.tex
\begin{proof}
	The corresponding reductions are all directly derived from the corresponding reductions for \ouhpg{} or \ouhcg{}.
	For \obtsg, the cost constraints can be relaxed to unit costs.
	For \omtsg, we can use the construction for unit cost \otsg{}, and extend it to a complete graph as in \Cref{ham:lem:walkAllVertices}, by choosing $0 < \alpha \leq 1$.
	This results in the \PSPACE-completeness of both problems, by setting the threshold to $|V|$ for \omtsg{} and to $1$ for \obtsg{}.
	The \PSPACE-completeness of both problems holds also for the trail and walk version of these problems as well as for directed and undirected graphs as proven in \Cref{sec:hampath}.
	
	Furthermore, the reduction for \olpg{} is the same for \ouhpg{}.
	Analogously, the reduction for \olcg{} is the same for \ouhcg{}.
	Only the threshold for both problems has to be set to $|V|$.
	Thus, both problems are \PSPACE-complete.
	This also holds for the directed version of the problems.
\end{proof}

% !TEX root=lipics-paper.tex
\section{Conclusion}
\label{sec:conclusion}
Graph Exploration and Treasure Hunt are interpretable as the online versions of classical $s$-$t$ path and Hamiltonian path problems.
We modeled the Graph Exploration and Treasure Hunt problems with an unlabeled map as online games between the online algorithm and the adversary to obtain decision versions of these problems.
Furthermore, we analyzed them from a complexity theoretic perspective and showed that nearly all are \PSPACE-complete.

It remains open whether the approximation of the discussed problems is \PSPACE-hard.
Another interesting question is the complexity of the existence of undirected online Hamiltonian trails and walks.
Additional path problems may be analyzed as well.
For example, one may find a path problem that is not directly reducible via the \ostpg{} because its constraints do not allow a standard $s$-$t$ path to be a solution.

Besides these open problems concerning graph exploration problems, the online version of other typical combinatorial problems may be analyzed such as \textsc{Partition}, \textsc{Scheduling} or \textsc{Matching}.

\newpage

\bibliography{bibliography}

\newpage

\end{document}